\documentclass[11pt,a4]{article}
\usepackage[utf8]{inputenc}
\usepackage{amsthm}
\usepackage{amsmath}
\usepackage{amssymb} 
\usepackage{thmtools} 
\usepackage{thm-restate}
\usepackage{xcolor,xspace}
\usepackage{nicefrac}
 \usepackage{enumerate}  

\usepackage{algpseudocode}
    \usepackage{algorithm}
    
     \usepackage{tikz}
    \tikzset{>=latex}  
    \usetikzlibrary{positioning,chains,fit,shapes,shapes.multipart,calc,topaths,decorations.pathreplacing,backgrounds}
    
\usepackage{hyperref}
\usepackage[nameinlink,capitalize]{cleveref}

\newcommand{\mypar}[1]{\paragraph{#1}}

\newtheorem{theorem}{Theorem}[section]
\newtheorem{definition}[theorem]{Definition}
\newtheorem{lemma}[theorem]{Lemma}
\newtheorem{proposition}[theorem]{Proposition}
\newtheorem{corollary}[theorem]{Corollary}

\newtheorem{observation}{Observation}
\newtheorem{remark}{Remark}

\newcommand{\alg}[1]{\textsc{#1}}
\newcommand{\eps}{\ensuremath{\epsilon}}

\newcommand{\logstar}{\ensuremath{\log^*}}
\newcommand{\myfrac}[2]{\nicefrac{#1}{#2}\,}
\newcommand{\CONGEST}{\ensuremath{\mathsf{CONGEST}}\xspace}
\newcommand{\LOCAL}{\ensuremath{\mathsf{LOCAL}}\xspace}
\newcommand{\set}[1]{\left\{#1\right\}}
\newcommand{\congest}{\CONGEST}

\DeclareMathOperator{\polylog}{polylog}
\DeclareMathOperator{\poly}{poly}
\DeclareMathOperator{\argmin}{argmin}
\DeclareMathOperator{\ID}{ID}

\begin{document}

\title{Distance-2 Coloring in the CONGEST Model}
\author{Magn\'us M. Halld\'orsson\thanks{ICE-TCS \& Department of Computer Science, Reykjavik University, Iceland. Partially supported by Icelandic Research Fund grant 174484-051.}  \and Fabian Kuhn \and Yannic Maus\footnote{Supported by the European Union's Horizon 2020 Research And  Innovation Programme under grant agreement no. 755839.}}

\date{\today}

\maketitle

\begin{abstract}
    We give efficient randomized and deterministic distributed algorithms for computing a distance-$2$ vertex coloring of a graph $G$ in the \CONGEST model. In particular, if $\Delta$ is the maximum degree of $G$, we show that there is a randomized \CONGEST model algorithm to compute a distance-$2$ coloring of $G$ with $\Delta^2+1$ colors in $O(\log\Delta\cdot\log n)$ rounds. Further if the number of colors is slightly increased to $(1+\eps)\Delta^2$ for some $\eps>1/\polylog n$, we show that it is even possible to compute a distance-$2$ coloring deterministically in $\polylog n$ time in the \CONGEST model. Finally, we give a $O(\Delta^2 + \log^* n)$-round deterministic \CONGEST algorithm to compute distance-$2$ coloring with $\Delta^2+1$ colors.
\end{abstract}

\section{Introduction} 
\label{sec:intro}
We study the distance-$2$ coloring problem in the standard distributed \CONGEST model. Given a graph $G=(V,E)$, in the distance-$2$ coloring problem on $G$ (in the following just called \emph{d2-coloring}), the objective is to assign a color $x_v$ to each node $v\in V$ such that any two nodes $u$ and $v$ at distance at most $2$ in $G$ are assigned different colors $x_u\neq x_v$. Equivalently, d2-coloring asks for a coloring of the nodes of $G$ such that for every $u\in V$, all the nodes in the set $\set{u}\cup N(u)$ (where $N(u)$ denotes the set of neighbors of $u$) are assigned distinct colors. Further note that d2-coloring on $G$ is also equivalent to the usual vertex coloring problem on the graph $G^2$, where $V(G^2)=V$ and there is an edge $\set{u,v}\in E(G^2)$ whenever $d_G(u,v)\leq 2$. 

The \CONGEST model is a standard synchronous message passing model~\cite{peleg00}. The graph on which we want to compute a coloring is also assumed to form the network topology. Each node $u\in V$ has a unique $O(\log n)$-bit identifier $\ID(u)$, where $n=|V|$ is the number of nodes of $G$. Time is divided into synchronous rounds and in each round, every node $u\in V$ of $G$ can do some arbitrary internal computation, send a (potentially different) message to each of its neighbors $v\in N(u)$, and receive the messages sent by its neighbors in the current round. If the content of the messages is not restricted, the model is known as the \LOCAL model~\cite{linial92,peleg00}. In the \CONGEST model, it is further assumed that each message consists of at most $O(\log n)$ bits. 

As our main result, we give an efficient $O(\log\Delta \log n)$-time randomized algorithm for d2-coloring $G$ with at most $\Delta^2+1$ colors, where $\Delta$ is the maximum degree of $G$. Further, we show that with slightly more colors, a similar result can also be achieved deterministically: We give a deterministic $\polylog n$-time algorithm to d2-color $G$ with $(1+\eps)\Delta^2$ colors for any $\eps>1/\polylog n$. 
Before discussing our results in more detail, we first discuss what is known for the corresponding coloring problems on $G$ and why it is challenging to transform \CONGEST algorithms to color $G$ into \CONGEST algorithms for d2-coloring. 

The distributed coloring problem is arguably the most intensively studied problem in the area of distributed graph algorithms and certainly also one of the most intensively studied problems in distributed computing more generally. The standard variant of the distributed coloring problem on $G$ asks for computing a vertex coloring with at most $\Delta+1$ colors. Note that such a coloring can be computed by a simple sequential greedy algorithm. In the following, we only discuss the work that is most relevant in the context of this paper, for a more detailed discussion of related work on distributed coloring, we refer to \cite{barenboimelkin_book,chang18_coloring,kuhn20_coloring}.

The $(\Delta+1)$-coloring problem was first studied in the parallel setting in the mid 1980s, where it was shown that the problem admits $O(\log n)$-time parallel solutions~\cite{alon86,luby86}. These algorithms immediately also lead to $O(\log n)$-round distributed algorithms, which even work in the \CONGEST model. In fact, even the following most simple algorithm $(\Delta+1)$-colors a graph $G$ in $O(\log n)$ rounds in the \CONGEST model: Initially all nodes are uncolored. The algorithm runs in synchronous phases, where in each phase, each still uncolored node $v$ chooses a uniform random color among its available colors (i.e., among the colors that have not already been picked by a neighbor) and $v$ keeps the color if no of its uncolored neighbors tries the same color at the same time~\cite{johansson99,BEPS12}.

Generally, the main focus in the literature on distributed coloring has been on the \LOCAL model, where by now the problem is understood relatively well. It was an important problem for a long time if there are similarly efficient deterministic algorithms for the distributed coloring problem (see, e.g., \cite{linial92,barenboimelkin_book,stoc17_complexity,derandomization_FOCS18}). This question was very recently resolved in a breakthrough paper by Rozho\v{n} and Ghaffari~\cite{RG19}, who showed that $(\Delta+1)$-coloring and many other important distributed graph problems have polylogarithmic-time deterministic algorithms in the \LOCAL model. The best randomized $(\Delta+1)$-coloring algorithm known 
in the \LOCAL model is by Chang, Li, and Pettie~\cite{chang18_coloring}, who show that the problem can be solved in time $\poly\log\log n$.\footnote{In \cite{chang18_coloring}, the complexity is given as $2^{O(\sqrt{\log\log n})}$. The improvement to $\poly\log\log n$ immediately follows from the recent paper by Rozho\v{n} and Ghaffari~\cite{RG19}. The same is true for the $n$-dependency in the \CONGEST model paper by Ghaffari~\cite{ghaffari19}, which is discussed below.} If the maximum degree $\Delta$ is small, the best known (deterministic) algorithm has a complexity of $O(\sqrt{\Delta\log\Delta}\cdot\log^*\Delta + \log^* n)$~\cite{fraigniaud16,BEG18}. We note that the $\log^* n$ term is known to be necessary due to a classic lower bound by Linial~\cite{linial92}. 

\paragraph{From coloring to d2-coloring.} While most existing distributed coloring algorithms were primarily developed for the \LOCAL model, several of them directly also work in the \CONGEST model (e.g., the ones in \cite{alon86,luby86,linial92,johansson99,Kuhn2006On,barenboim10,BEK15,barenboim15,BEG18,kuhn20_coloring}). There is also some recent work, which explicitly studies distributed coloring in the \CONGEST model. In \cite{ghaffari19}, Ghaffari gives a randomized $(\Delta+1)$-coloring algorithm that runs in $O(\log\Delta) + \poly\log\log n$ rounds in the \CONGEST model. For the \CONGEST model, this is the first improvement over the simple randomized $O(\log n)$-round algorithms from the 1980s. Further, in another recent paper~\cite{det_congest_coloring}, by building on the recent breakthrough in the \LOCAL model~\cite{RG19}, it is shown that it is also possible to deterministically compute a $(\Delta+1)$-coloring in $\polylog n$ time in the \CONGEST model.

In the \LOCAL model, a single communication round on $G^2$ can be simulated in $2$ rounds on $G$ and therefore the distributed coloring problem on $G^2$ is at most as hard as the corresponding problem on $G$.\footnote{Note that not every graph $H$ is the square $G^2$ of some graph $G$ and thus, the coloring problem on $G^2$ might be easier than the coloring problem on $G$.} In the \CONGEST model, the situation changes drastically and it is no longer generally true that a \CONGEST algorithm on $G^2$ can be run at a small additional cost on the underlying graph $G$. In general, simulating a single \CONGEST round on $G^2$ requires $\Omega(\Delta)$ \CONGEST rounds on $G$. Note that even the very simple algorithm where each node picks a random available color cannot be efficiently used for d2-coloring as it is in general not possible to keep track of the set of colors chosen by some $2$-hop neighbor in time $o(\Delta)$. In some sense, our main technical contribution is an efficient randomized \CONGEST algorithm (on $G$) that implements this basic idea of iteratively trying a random color until all nodes are colored.

\paragraph{Why d2-coloring?} Distributed d2-coloring is an interesting and important problem for several reasons. It is fundamental in wireless networking, where nodes with common neighbors \emph{interfere} with each other. Computing a frequency assignment such that nodes with the same frequency do not interfere with each other therefore corresponds to computing a d2-coloring of the communication graph \cite{KMR01}. Computing a coloring in a more powerful model (\congest) than it would be used in (wireless channels) is in line with current trends towards separation of control plane and data plane in networking. The d2-coloring problem also occurs naturally when single-round randomized algorithms are derandomized using the method of conditional expectation~\cite{derandomization_FOCS18}. Further, d2-coloring forms the essential part of \emph{strong coloring} hypergraphs, where nodes contained in the same hyperedge must be colored differently. One natural setting is when the nodes form a bipartite graph, with, say, ``task'' nodes on one side and ``resource'' nodes on the other side. We want to color the task nodes so that nodes using the same resource receive different colors.

Finally, we can also view d2-coloring and other problems on $G^2$ as a way of studying \emph{communication capacity constraints} on nodes, where communication must go through intermediate relays. In fact, d2-coloring in \CONGEST is of special interest as it appears to lie at the edge of what is computable efficiently, i.e., in polylogarithmic time. Many closely related problems are either very easy or quite hard. The \emph{distance-$k$ maximal independent set} problem can easily be solved in $O(k\log n)$ time using Luby's algorithm~\cite{alon86,luby86}. The distance-3 coloring problem, however, appears to be hard. There is a simple reduction from the hardness of the $2$-party set disjointness problem~\cite{kalyanasundaram92,razborov92} to show that the closely related problem of verifying whether a given distance-$3$ coloring is valid requires $\Omega(\Delta)$ rounds, even on graphs where $\Delta=\Theta(n)$ (just think of a tree consisting of an edge $\set{a,b}$ and with $(n-2)/2$ leaf nodes attached to both $a$ and $b$). In fact, the classic set disjointness lower bound proof of Razborov~\cite{razborov92} implies that even verifying validity of a uniformly random coloring is hard.

\subsection{Contributions}

We provide different \CONGEST model algorithms to compute a d2-coloring of a given $n$-node graph $G=(V,E)$. If $\Delta$ is the maximum degree of $G$, the maximum degree of any node in $G^2$ is at most $\Delta + \Delta\cdot(\Delta-1)=\Delta^2$. As a natural analog to studying $(\Delta+1)$-coloring on $G$, we therefore study the problem of computing a d2-coloring with $\Delta^2+1$ colors. Although, there are extremely simple $O(\log n)$-time randomized algorithms for $(\Delta+1)$-coloring $G$, transforming similar ideas to d2-coloring turns out to be quite challenging. Our main technical contribution is an efficient randomized algorithm to d2-color $G$ with $\Delta^2+1$ colors. 

\begin{theorem}
\label{thm:d2ColoringRand}
There is a randomized \CONGEST algorithm that d2-colors a graph with $\Delta^2+1$ colors in $O(\log \Delta \log n)$ rounds, with high probability.
\end{theorem}

We outline the key ideas and challenges involved at the start of Sec.~\ref{sec:randAlg}. 

In addition to the randomized algorithm for computing a d2-coloring, we also provide two deterministic algorithms for the problem. The first one is obtained by a relatively simple adaptation of an $O(\Delta+\log^*n)$-time $(\Delta+1)$-coloring \CONGEST algorithm on $G$ to the d2-coloring setting~\cite{BEG18}.

\begin{theorem}
\label{thm:d2ColoringDelta}
There is a deterministic \CONGEST algorithm that d2-colors a graph with $\Delta^2+1$ colors in $O(\Delta^2+\logstar n)$ rounds.
\end{theorem}

Our second deterministic algorithm is more involved. From a high-level view, it uses ideas similar to several recent \CONGEST results ~\cite{det_congest_coloring,CPS17,DKM19,DISC18_DomSet}: 
With the algorithm of \cite{RG19}, one decomposes the graph into clusters of $\polylog n$ diameter that the problem can essentially be solved separately on each cluster (incurring a polylogarithmic overhead). On each cluster, one then uses the method of conditional expectation to efficiently derandomize a simple zero-round randomized algorithm. Unlike the algorithms in \cite{det_congest_coloring,CPS17,DKM19,DISC18_DomSet}, we do not use this general strategy to directly solve (a part of) the problem at hand (d2-coloring in our case). Instead, we apply the above strategy to implement a variant of the splitting problem discussed in \cite{BambergerGKMU19,stoc17_complexity}. By applying the splitting problem recursively, we partition the nodes $V$ into $\Delta/\polylog n$ parts such that a) we can use disjoint color palettes for the different parts, and b) we can efficiently simulate \CONGEST algorithms on $G^2$ on each of the parts (and these \CONGEST simulations can also efficiently be run in parallel on all the parts). By using slightly more colors, we can then also compute a d2-coloring in $\polylog n$ time deterministically.

\begin{theorem}[Simplified]
\label{thm:G2withSplitting}
For any fixed constant $\eps>0$, there is a deterministic \CONGEST algorithm that d2-colors a graph with $(1+\eps)\Delta^2$ colors in $\polylog n$ rounds.
\end{theorem}


The remainder of the paper is structured as follows. In \Cref{sec:randAlg}, we present our randomized algorithm and prove \Cref{thm:d2ColoringRand}, our main technical result. In \Cref{S:det-diam}, we present our deterministic algorithms, proving \Cref{thm:d2ColoringDelta,thm:G2withSplitting}. Note that because of space restrictions, many of the proofs appear in an appendix. 


\section{Randomized Algorithm}
\label{sec:randAlg}
We give randomized {\congest} algorithms that form a d2-coloring using $\Delta^2+1$ colors. We use the prominent space at the beginning of the section to introduce notation that we use frequently throughout the proofs in this section.

\paragraph{Notation}
The \emph{palette} of available colors is $[\Delta^2] = \{0,1,2,\ldots, \Delta^2\}$. 
The neighbors in $G$ of a node are called \emph{immediate neighbors}, while the neighbors in $G^2$ are \emph{d2-neighbors}.
For a (sub)graph $K$, let $N_K(v)$ denote the set of neighbors of $v$ in $K$, and let $K[v] = K[N_K(v)]$ denote the subgraph induced by these neighbors.
A node is \emph{live} until it becomes \emph{colored}.

A node has \emph{slack $q$} if the number of colors of d2-neighbors plus the number of live d2-neighbors is $\Delta^2+1-q$. In other words, a node has slack $q$ if its palette size is an additive $q$ larger than the number of its uncolored $d2$-neighbors. 
The \emph{leeway} of a node is its slack plus the number of live d2-neighbors; i.e., it is the number of colors from the palette that are not used among its d2-neighbors.
During our algorithms nodes do not know their leeway and we only use the notion for the analysis.

When we state that an event holds \emph{w.h.p.} (with high probability), we mean that for any $c > 0$, we can choose the constants involved so that the event holds with probability $1 -O(n^{-c})$. 

\subsection{Overview : Coloring 'With a Little Help From My Friends'}

As explained in the introduction, the simple approach for coloring $G$ -- for each node to guess a random color that is currently not used among any of its neighbors -- fails for d2-coloring because the nodes do not have enough bandwidth to learn the colors of their d2-neighbors. 

Instead, nodes can certainly \emph{try} a random color from the whole palette. The node's immediate neighbors can maintain their immediate neighbors colors, and thus can answer if a certain color conflicts with the current coloring (or other colors being tried).
This works well in the beginning, until most of the node's neighbors are colored.
If the palette has $(1+\epsilon)\Delta^2$ colors, then this approach alone succeeds in $O(\log_{1/\epsilon} n)$ rounds, but for a $\Delta^2+1$-coloring, we must be more parsimonious. 
If each neighborhood is sparse, then the first round will result in many d2-neighbors successfully using the same color. 
This offers us then the same slack as if we had a larger palette in advance, as proved formally by Elkin, Pettie and Su \cite{EPS15}, resulting in the same logarithmic time complexity.
The challenge is then to deal with dense neighborhoods, of varying degrees of sparsity, defined formally for each node as the average non-degree of the subgraph induced by its neighborhood in $G^2$.
We tackle this with the algorithm \alg{Reduce}, that successfully colors all nodes in a given range of color slack (and by extension, sparsity range).

The basic idea behind the \alg{Reduce} algorithm is to have the colored nodes "help" the \emph{live} (i.e., yet uncolored) nodes by checking random colors on their neighborhoods.
%
We can obtain some intuition from the densest case: a $\Delta^2+1$-clique (in $G^2$).
We can recruit the colored nodes to help the live nodes guess a color: if it succeeds for the colored node, it will also succeed for the live node. Each of the $\ell$ live nodes can be allocated approximately $\Delta^2/\ell$ colored node helpers, and in each round, with constant probability, one of them successfully guesses a valid color. This reduces the number of live nodes by a constant factor, leading to a $O(\log n)$ time complexity.

The challenge in more general settings is that the nodes no longer have identical (closed) d2-neighborhoods, so a successful guess for one node does not immediately translate to a successful color for another node. To this end, we must deal with two types of errors.
A \emph{false positive} is a color that works for a colored node $w$ but not for its live d2-neighbor $v$, while a \emph{false negative} is a color that fails for the colored node but succeeds for the live node. It is not hard to conceive of instances where there are no true positives.

The key to resolving this is to use only advice from nodes that have highly \emph{similar} d2-neighborhoods.
This is captured as a relationship on the nodes: the similarity graph $H\subseteq G^2$. We also use another similarity graph $\hat{H}\subseteq G^2$, with a higher threshold for similarity (in terms of number of common d2-neighbors).
To combat false negatives, we also try colors of similar nodes that are not d2-neighbors of the live node but have a common (and similar) $H$-neighbor with the live node, i.e., we try the colors of nodes in $N_{H^2}(v) \setminus N_{G^2}(v)$.

Additional challenges and pitfalls abound. We must carefully balance the need for progress with the load constraints on each node or edge. Especially, the efforts of the live nodes are a precious resource, but we must allow for their distribution to be decidedly non-random.  
In addition, there are differences between working on 2-paths in $G$ and on edges in $G^2$: there can be multiple 2-paths between d2-neighbors. This can confound seemingly simple tasks such as picking a random d2-neighbor.

Once bounds on sparsity and slack drop below logarithmic, concentration results fail to hold. 
Finishing up becomes the bottleneck of the whole algorithm.
For this, we introduce an improved algorithm. The key is that there is now sufficient bandwidth for the remaining live nodes to learn the \emph{complement} of the set of colors of their d2-neighbors: the colors that they \emph{don't use}. Though there is no obvious way for them to discover that alone, they can again get help from the colored node in tallying the colors used. This becomes a different problem of outsourcing and load-balancing, but one that is aided by the extreme denseness of the parts of the graph that are not yet fully colored. We explain this in more detail in Sec.~\ref{ssec:improved}.
Once the palette is known, the rest of the algorithm is like for the basic randomized algorithm for coloring $G$, since the nodes can maintain an up-to-date view of the colors of their d2-neighbors.

\subsection{Algorithm Description}

We now outline our top-level algorithm, followed by the main routine, \alg{Reduce} and details on the implementation. 

Recall that a node $v$ \emph{trying} a color means that it sends the color to all its immediate neighbors, who then report back if they or any of their neighbors were using (or proposing) that color.
If all answers are negative, then $v$ adopts the color. 

In what follows, $c_0$, $c_1$ are constants satisfying $c_0 \le 3e/c_1$, $c_1 \le 1/(402 e^3)$. Also, $c_2$ is a sufficiently large constant needed for concentration.

\begin{quote}
   \textbf{Algorithm} \alg{d2-Color}

   0. If $\Delta^2 < c_2 \log n$ then \alg{Deterministic-d2Color($G$)}; halt \\
   1. Form the similarity graphs $H$ and $\hat{H}$  \hspace{1cm }\textit{// Initial Phase} \\
   2. repeat $c_0 \log n$ times: \\
   \hspace*{2em} Each live node picks a random color and \emph{tries} it. \\
   3. for ($\tau \leftarrow c_1 \Delta^2$; $\tau > c_2 \log n$; $\tau \leftarrow \tau/2$) \hspace{3cm }\textit{// Main Phase}\\
\hspace*{2em}      \alg{Reduce}($2 \tau$, $\tau$) \\
    4. \alg{Reduce}($c_2 \log n$, 1)\hspace{5.75cm}\textit{// Final Phase}
\end{quote}

For low-degree graphs, we use in Step 0 the deterministic algorithm from Sec.~\ref{ssec:sumDeltaDeta}. The similarity graphs $H$ and $\hat{H}$ that are constructed in Step 1 are  used later  (in \alg{Reduce}) to decide which nodes assists whom. The point of Step 2 is to reduce the initial number of live nodes down to a small fraction of each neighborhood. We can then apply the main algorithm, \alg{Reduce}, to progressively reduce the leeway of live nodes (by coloring them or their neighbors).

We let $c_3$ be a sufficiently large constant to be determined. 
\medskip

   \textbf{Algorithm} \alg{Reduce}($\phi$, $\tau$)
   
    \emph{Precondition}: Live nodes have leeway less than $\phi$, where $c_2\log n \le \phi \le c_1\Delta^2$
    
    \emph{Postcondition}: Live nodes have leeway less than $\tau$

\begin{quote}
   Each node $u$ selects a multiset $R_u$ of $\rho \doteq c_3 (\phi/\tau)^2 \log n$ random $H$-neighbors (with replacement) \\
    Repeat $\rho$ times: \\
    \hspace*{2em} Each live node is \emph{active} independently with probability $\tau/(8\phi)$ \\
    \hspace*{2em} \alg{Reduce-Phase}($\phi,\tau$) 
\end{quote}

The selection of random $H$-neighbors needs care and is treated in the following subsection.
\alg{Reduce}($\phi,\tau$) ensures that all nodes with a certain range of leeway get colored, which implicitly ensures that the number of live nodes in each neighborhood goes down as well. To avoid too much competition between live nodes, only a fraction of them participate in any given phase.
    
   \textbf{Algorithm} \alg{Reduce-Phase}($\phi$, $\tau$)
\begin{enumerate}
\itemsep 0em 
  \item Each active live node $v$ sends a query across each 2-path to $\hat{H}$-neighbors independently with probability $1/(6000\phi)$.
  \item The recipient $u$ of a query $(v,u)$ verifies that there is only a single 2-path from $v$, and otherwise drops the message.
  \item $u$ picks a random color $\hat{c}$ different from its own and checks if it is used by any of its $H$-neighbors. If not, it sends the color back to $v$ as a proposal.
  \item $u$ also forwards the query to the next uniformly random $H$-neighbor $w$ from its list $R_u$, with $w$ appended to the query.
  \item Upon receipt of query $(v,u,w)$, node $w$ checks if $v$ is a d2-neighbor; if not, the color $c(w)$ of $w$ is sent to $v$ (through $u$).
  \item The active live node $v$ tries a color chosen uniformly random among the proposed colors (if any).
  \end{enumerate}
At each step along the way, a node receiving multiple queries selects one of them at random and drops the others. This can only occur after both rounds of Step 1, first round of Step 2, or second round of Step 4.

\alg{Reduce-Phase} ensures that all active live nodes (with leeway between $\tau$ and $\phi$) get colored with a "constant" probability (i.e., a constant times $\tau/\phi$). This is achieved by each live node recruiting a large subset of its similar d2-neighbors to try random colors (in Step 3). This is a probabilistic filter that reduces the workload of the live nodes. These neighbors also check the colors of their neighbors (in Step 5) to see if those might be suitable for the live node. The key idea is that one of these forms of assistance is likely to be successful, and that it is possible to share the load effectively.

\mypar{Implementation}
Additional details for specific steps of \alg{Reduce-Phase}:

\emph{Step 1:}
When sending a query along 2-paths in Step 1, the node $v$ simply asks its immediate neighbors to send the queries to all of their immediate neighbors that are $H$-neighbors of $v$, with the given probability.

\emph{Step 2:}
Verifying that there is only a single path from $v$ is achieved by asking $u$'s immediate neighbors how many are neighbors of $v$. 

\emph{Step 3:} Checking if a color is used by an $H$-neighbor is identical to trying a color, but having the immediate neighbors only taking into account the colors of $u$'s $H$-neighbors.

We detail in the following subsection how $R_u$, the collection of random $H$-neighbors, is generated in Step 4 in time proportional to its size. 
Steps 5 and 6 of \alg{Reduce-Phase} are straightforward to implement. We note that a query from a live node $v$ maintains a full routing path to $v$, so getting a proposal back to $v$ is simple.

\mypar{Complexity}
For low-degree graphs ($\Delta^2 = O(\log n)$), we use the deterministic algorithm of \Cref{thm:d2ColoringDelta}, which runs in $O(\Delta^2 + \log^* n) = O(\log n)$ rounds.
We show in the next subsection that the first step of \alg{d2-Color} takes $O(\log n)$ rounds, w.h.p.
The second step clearly takes $\Theta(\log n)$ rounds.

The procedure \alg{Reduce-Phase} takes 23 rounds, or 2 (Step 1), 4 (Step 2), 4 (Step 3), 2 (Step 4), 6 (Step 5), and 5 (Step 6, including the notification of a new color).
Thus, the round complexity of \alg{Reduce} (including the time to generate $R_u$) is proportional to the number of iterations of the loop, or $O((\phi/\tau)^2 \log n)$.
It follows that all the steps of \alg{d2-Color} run in $O(\log n)$ time, except the last step, i.e., $\alg{Reduce}(c_2\log n,1)$, that requires $O(\log^3 n)$ time.
Since we also show that at the end every vertex is colored, w.h.p., we obtain the following result.

\begin{corollary}
  There is a randomized {\congest} algorithm to d2-color with $\Delta^2+1$ color in $O(\log^3 n)$ rounds, w.h.p.
  \label{C:first-rand-result}
\end{corollary}

Outline of the rest of this section: In Sec.~\ref{ssec:similarity} we describe the remaining supporting steps of the algorithms, in Sec.~\ref{ssec:dense} we derive key structural properties of non-sparse graphs, and in Sec.~\ref{ssec:correctness} we prove the correctness of the algorithms. 
Then, in Sec.~\ref{ssec:improved}, we present an improved algorithm that replaces the last step of \alg{d2Color} to reduce the overall time complexity to $O(\log \Delta \log n)$, giving our main result, Thm.~\ref{thm:d2ColoringRand}.

\subsection{Support Functions : Similarity Graphs and Random Neighbor Selection}
\label{ssec:similarity}
We describe here in more details the support tools and property used in our algorithm.
This includes the formation of the similarity graph, and the selection of random d2-neighbors. 

\mypar{Forming the similarity graphs} 
%
%
We form the \emph{similarity graph} $H = H_{2/3}$ on the nodes of $V = V(G)$, where nodes are adjacent only if they are d2-neighbors and have at least $2\Delta^2/3$ d2-neighbors in common.
This is implemented in the sense that each node knows:
a) whether it is a node in $H$, and
b) which of its immediate neighbors are adjacent in $H$.
If a node has no neighbor in $H$, we consider it to be not contained in $H$.

When $\Delta^2 = O(\log n)$, each node can gather its set of d2-neighbors and forward it to its immediate neighbors in $O(\log n)$ rounds. The immediate neighbors can then determine which of its immediate neighbors share
at least $2\Delta^2/3$ common d2-neighbors, which defines $H$.
We focus from now on the case that $\Delta^2 \ge c_{10}\log n$, for appropriate constant $c_{10}$.

To form $H$, each node chooses independently with probability $p = c_{10}(\log n)/\Delta^2$ whether to enter a set $S$. Nodes in $S$ inform their d2-neighbors of that fact. For each node $v$, let $S_v$ be the set of d2-neighbors in $S$. W.h.p., $|S_v| = O(\log n)$ (by Prop.~\ref{P:chernoff}). Each node $v$ informs its immediate neighbors of $S_v$, by pipelining in $O(\log n)$ steps.
Note that a node $w$ can now determine the intersection $S_v \cap S_u$, for its immediate neighbors $v$ and $u$.
Now, d2-neighbors $u$, $v$ are $H$-neighbors iff $|S_v \cap S_u| \ge \myfrac{5}{6} c_{10} \log n$. 
\begin{theorem}
Let $k \in \{3,6\}$.
Let $u,v$ be d2-neighbors. 
If $(u,v) \in H_{1-1/k}$ (i.e., if $|S_v \cap S_u| \ge (1-1/(2k)) c_{10} \log n$), then they share at least $(1-1/k) \Delta^2$ common d2-neighbors, w.h.p.,
while if $(u,v) \not\in H$, then they share fewer than $(1-1/(4k)) \Delta^2$ common d2-neighbors, w.h.p.
\label{T:similarity}
\end{theorem}
The proof uses Chernoff bounds and is deferred to the appendix. 

We also form the graph $\hat{H} = H_{5/6}$ in an equivalent manner.
For $H_{1-1/k}$, the condition used by the algorithm becomes $|S_v \cap S_u| \ge (1-k/2)c_{10} \log n$ and the case of when $(u,v) \not\in H_{1-1/k}$ is when they share fewer than $(1-k/4)\Delta^2$ common neighbors.

\mypar{Selecting random $H$-neighbors}
We detail how the multiset $R_u$ of uniformly random $H$-neighbors is form, at the start of \alg{Reduce}.
We repeat the following procedure $\rho$ times, to create a list $R_u$ of $\rho$ random $H$-neighbors at each node:
Each node $u$ that receives a query creates a $4\log n$-bit random string $b_u$, and transmits it to all its immediate neighbors. Each node $w$ also picks a $4\log n$-bit random string $r_w$ and sends to immediate neighbors. Now, each immediate node $u'$ computes the bitwise XOR $x_{uw}$ of each string $b_u$ and each string $r_w$ that it receives, where $u$ and $w$ are $H$-neighbors. It forwards $r_w$ to $u$ if and only if the first $2\log \Delta - c_{11} \log\log n$ bits of $x_{uw}$ are zero. 
The node $u$ then selects the $H$-neighbor $w$ with the smallest XORed string $b_u \oplus r_w$.

\begin{lemma}
A multiset $R_u$ of independent uniformly random $H$-neighbors of node $u$ can be generated in $O(|R_u|+\log n)$ rounds. 
\label{L:rand-nbor}
\end{lemma}

\subsection{Properties of Dense Subgraphs}
\label{ssec:dense}

The example of the clique at the start of this subsection shows that dense subgraphs have the advantage that the views of the nodes are homogeneous. The advantage of sparse subgraphs is that they will invariable have slack, as shown by the following result of \cite{EPS15}.

We frequently work with nodes that are both sparse enough and of small enough leeway. 

\begin{definition}
A node $v$ is \emph{$\zeta$-sparse} (or \emph{has sparsity} $\zeta$) if $G^2[v]$ contains $\binom{\Delta^2}{2} - \Delta^2\cdot \zeta$ edges. 
$v$ is \emph{solid} if it has leeway $\phi \le c_1 \Delta^2$ and sparsity $\zeta \le 4 e^3 \phi$. 
\label{D:solid}
\end{definition}

Sparsity is a rational number in the range $0$ to $(\Delta^2-1)/2$ that is fixed throughout. Leeway is a decreasing property of the current partial coloring. Thus, once a node becomes solid, it stays solid throughout the algorithm.
Elkin, Pettie and Su \cite{EPS15} formalized the connection between the two properties.


\begin{proposition}[\cite{EPS15}, Lemma 3.1]
Let $v$ be a vertex of sparsity $\zeta$ and let $Z$ be the slack of $v$ after the first round of \alg{d2-Color}. Then, 
 $\Pr[Z \le \zeta/(4 e^3)] \le e^{-\Omega(\zeta)}$.
\label{P:sparsity}
\end{proposition}

We require the constant $c_2$ to be such that if $\zeta \ge c_2\log n$, then the contraposition of Prop.~\ref{P:sparsity} yields that $Z \ge \zeta/(4e^3)$, w.h.p. 

We derive some of the essential features of low-sparsity neighborhoods: almost all d2-neighbors are also $H$-neighbors, and almost all neighbors in $H^2$ are also d2-neighbors.
The first part applies both to $H = H_{2/3}$ and $\hat{H} = H_{5/6}$. 

\begin{lemma}
Let $v$ be a node of sparsity $\zeta$.
Then,
\begin{enumerate}
    \item $v$ has at least $\Delta^2 - 8\zeta/k -4/k$ neighbors in $H_{1-k}$, and
    \item The number of nodes that are within distance 2 of $v$ in $H$ but are not d2-neighbors of $v$ is
$|N_{H^2}(v) \setminus N_{G^2(v)}| \le 6\zeta$.
\end{enumerate}
\label{L:h-degree}
\end{lemma}

\begin{observation}
Every live node is solid after Step 1 of \alg{d2-Color}, w.h.p.
\label{O:sparse}
\end{observation}

Let $H'$ denote the subgraph of $\hat{H}[v]$ induced by nodes with a single 2-path to $v$.
Let $deg_H(u)$ denote the number of $H$-neighbors of node $u$.
Solid nodes have many neighbors in $H'$, and its neighbors have many $H$-neighbors.
\begin{lemma}
Let $v$ be a solid node.
Then,
\begin{enumerate} 
  \item $v$ has at least $\Delta^2/2$ $H'$-neighbors.
  \item Every $\hat{H}$-neighbor of $v$ has at least $\Delta^2/3$ $H$-neighbors.
  \item The degree sum in $N_{H'}(v)$ is bounded below by
      \[ \sum_{u \in N_{H'}(v)} deg_H(u) \ge |N_{H'}(v)| (\Delta^2 - c_8\phi), \]
     for constant $c_8 \le 4000$.
\end{enumerate}
\label{L:H-neighbors}
\end{lemma}


\subsection{Correctness of Reduce}
\label{ssec:correctness}
Recall that the leeway of a node counts the number of colors of the palette that are not used among its neighbors. It counts both the number of uncolored nodes and the color slack that follows from the node being solid (by Obs.~\ref{O:sparse}). 
During this whole section we assume that all nodes are solid, and that the precondition of \alg{Reduce} is satisfied, i.e., live node's leeway is at most $\phi$ with $\phi\geq c_2\log n$ for a large enough constant $c_2$.
We also assume that similarity graphs $H$ and $\hat{H}$ are correctly constructed, in the sense of Thm.~\ref{T:similarity}. 
All statements in this section are conditioned on these events.

The algorithm is based on each live node sending out a host of queries, to random neighbors in $\hat{H}$, and through them to random $H$-neighbors. We argue that each query has a non-trivial probability of leading to the live node becoming colored.
We say that a given (randomly generated) query \emph{survives} if it is not dropped in any of Steps 1 - 5 due to congestion. This does not account for the outcome of the color tries of Steps 3 and 5.

Missing proofs are given in Sec.~\ref{app:improved}

\begin{lemma} Let $v$ be an active live node. 
Any given query sent from $v$ towards a node $w$ via a node $u\in H'\subseteq \hat{H}[v]$ survives with constant probability at least $c_6 \ge 1/7$, independent of the path that the query takes.
\label{L:survival}
\end{lemma}

The following progress lemma is the core of our correctness argument.
A color is \emph{$v$-good} if it is not used among the d2-neighbors of $v$ at the start of \alg{Reduce-Phase}.

\begin{lemma}
Let $v$ be an active live node at the start of \alg{Reduce-Phase}($\phi,\tau$) and let $\sigma$ be a $v$-good color. The probability that $\sigma$ is proposed to $v$ is at least $c_6/(24000\phi)$.
\label{L:sigma}
\end{lemma}
\begin{proof}
We first analyze a hypothetical situation where no queries are dropped.

Suppose $v$ generates a query $Q=(v,u)$ towards a $H'$-neighbor $u$.
%
We consider two cases, depending on whether the color $\sigma$ appears on an $H$-neighbor of $u$ (at the start of \alg{Reduce-Phase}). We claim that in either case, $\sigma$ gets proposed to $v$ with probability at least $1/\Delta^2$.

\textit{Case 1, $\sigma$ is used by an $H$-neighbor  of $u$: } Let $w$ be an $H$-neighbor of $u$ with color $\sigma$. Then $w$ is not a d2-neighbor of $v$, since $\sigma$ is $v$-good.
With probability at least $1/\Delta^2$, $u$ forwards the query to $w$, who then sends it as proposal to $v$.

\textit{Case 2, $\sigma$ does not appear among $u$'s $H$-neighbors:} Then with probability $1/\Delta^2$, $u$ will pick $\sigma$ as $\hat{c}$, try it successfully, and propose it to $v$.

Thus, in both cases the probability that a query $Q$ leads to $\sigma$ being proposed to $v$ is at least $1/\Delta^2$, given that $Q$ was generated and that it survives.
The probability that $Q=(v,u)$ is generated is $1/(6000\phi)$, and it is independent of it leading to a particular color.
Thus, the probability that a query $Q$ leads to $\sigma$ being proposed to $v$ is at least $1/(6000\phi\Delta^2)$, given that $Q$ survives. 


We now consider the event that none of $v$'s queries result in a proposal of $\sigma$. Since we are in the setting where no queries are dropped, 
the events for different intermediate nodes $u$ are independent.
Recall that by Lemma \ref{L:H-neighbors}(1), $|N_{H'}(v)| \ge \Delta^2/2$.
Thus, the probability that none of $v$'s queries result in a proposal of $\sigma$ is at most
\[ (1-1/(6000\phi\Delta^2))^{\Delta^2/2} \le  e^{-1/(12000\phi)} \le 1-1/(24000\phi)\ , \]
using the inequality $e^{-x} \le 1-x/2$, for $x \le 1/2$.
That is, the probability that there is a query $Q$ in which $\sigma$ is proposed to $v$ is at least $1/(24000\phi)$, under our assumption that no queries are dropped.

By Lemma \ref{L:survival}, a query survives with probability at least $c_6 \ge 1/7$, independent of the path it takes, and thus independent of the color it leads to.
Thus, the probability that $\sigma$ gets proposed to $v$ (via some $u \in H'[v]$) is at least 
$c_6/(24000\phi)$.
\end{proof}

\begin{lemma}
An active live node receives at most one proposal in expectation. This holds even in the setting where no queries are dropped. 
\label{L:proposal-expected}
\end{lemma}

\begin{lemma}
  Let $v$ be an active live node.
  Conditioned on the event that a particular color is proposed to $v$, the probability that $v$ \emph{tries} the color is at least $c_9 \ge 1/6$.
\label{L:proposal-tried}
\end{lemma}

\begin{lemma}
An active live node with leeway at least $\tau$ at the start of \alg{Reduce-Phase}($\phi,\tau$) becomes colored with probability $c_7 \tau/\phi$, for some constant $c_7 > 0$.
\label{L:progress}
\end{lemma}
\begin{proof}
Let $v$ denote the active live node and let $\Psi$ denote the set of $v$-good colors. By the leeway bound, $|\Psi| \ge \tau$. 
Let $\Phi$ be the multiset of colors proposed to (active) live d2-neighbors of $v$.
There are at most $\phi$ live d2-neighbors, and the expected fraction of them that are active is $\tau/(8\phi)$. Each active live node receives expected at most 1 proposals. Hence, the expected size of $\Phi$ is at most $\phi \cdot \tau/(8\phi) \cdot 1 = \tau/8$.
Let $A$ be the event that $\Phi$ is of size at most $\tau/4$. By Markov's inequality, $A$ holds with probability at least $\Pr[A] \ge 1-1/2 = 1/2$.

For a color $\sigma \in \Psi$, let $p_\sigma$ be the probability that $\sigma$ is proposed to some active live d2-neighbor of $v$. This dominates the probability that a d2-neighbor of $v$ actually \emph{tries} $\sigma$. 
Let $\Psi' = \{\sigma \in \Psi : p_{\sigma} \ge 1/2\}$. 
Note that $\sum_{\sigma \in \Psi} p_{\sigma} \le |\Phi| \le \tau/4$, assuming $A$ holds. 
On the other hand, the sum is at least $\sum_{\sigma \in \Psi'} p_{z} \ge \sum_{\sigma \in \Psi'} 1/2 = |\Psi'|/2$. Thus, $|\Psi'|/2 \le \tau/4$, or $|\Psi'| \le \tau/2$, assuming $A$.

Let $\hat{\Psi} = \Psi \setminus \Psi'$ and let $\sigma \in \hat{\Psi}$. Let $B_\sigma$ be the event that $v$ tries $\sigma$ while no d2-neighbor of $v$ receives a proposal of $\sigma$. Observe that the events for different $\sigma$ are independent.
By Lemma \ref{L:sigma} that $\sigma$ is proposed to $v$ is at least $c_6/(24000\phi)$ and by Lemma \ref{L:proposal-tried}, the probability that it gets tried is at least $c_9 \ge 1/6$. 
Assuming $A$ holds, $|\hat{\Psi}| \ge \tau/2$.
For $\sigma \in \hat{\Psi}$, the probability that no d2-neighbor of $v$ receives a proposal of $\sigma$ is at least $1/2$, assuming $A$.
Thus, for $\sigma \in \hat{\Psi}$,
\[ \Pr[B_\sigma] \ge \frac{c_6}{24000\phi} \cdot \frac{1}{6} \cdot \frac{1}{2} \ge \frac{1}{300000 \phi}\ , \]
Since $v$ tries only one color, the events $B_\sigma$ are disjoint.
Let $B = \cap_{\sigma \in \hat{\Psi}} B_\sigma$.
Then, $\Pr[B] = |\hat{\Psi}|/(300000\phi)$. 
Assuming $A$, $\Pr[B|A] \ge (\tau/2)(1/300000\phi) = \tau/(600000\phi)$.
Now, if $B$ holds, then $v$ becomes colored.
This happens with probability at least
$\Pr[B] \ge \Pr[B \cap A] = \Pr[A] \Pr[B|A] = 1/2 \cdot \tau/(600000\phi) = \tau/(1200000\phi)$.
\end{proof}

Since the algorithm performs $O(\log n)$ phases, and each node of leeway at least $\tau$ is colored in each phase with a constant probability, the algorithm either properly colors the node or decreases its leeway below $\tau$.

\begin{theorem}
All live nodes are of leeway less than $\tau$ after the call to \alg{Reduce}($\phi,\tau$), w.h.p.
\label{T:reduce-correct}
\end{theorem}

\begin{proof}
Set $c_3 = 32/c_7$. 
Let $v$ be a live node of leeway at least $\tau$ at the start of $\alg{Reduce}$.
With probability $\tau/(8\phi)$, $v$ is active in a given phase, and with
probability at least $c_7 \tau/\phi$, $v$ becomes colored in a given phase where it is active, by Lemma \ref{L:progress}. Thus, the probability that it remains live after all $\rho = c_3 (\phi/\tau)^2 \log n$ phases is at most 
$(1-c_7\cdot(\tau/\phi)^2/8)^{c_3 (\phi/\tau)^2 \log n} \le e^{-4\log n} \le n^{-4}$.
The probability that some such node remains uncolored is at most $n^{-3}$.
\end{proof}

Observe that after \alg{Reduce}($\phi$,1), all nodes are colored, w.h.p., since a live node always has leeway at least 1. 
Corollary \ref{C:first-rand-result} follows.

\subsection{Algorithm with Improved Final Phase}
\label{ssec:improved}

We now give an improved algorithm for d2-coloring that uses $\Delta^2+1$ colors and runs in $O(\log \Delta\log n)$ rounds. We assume $\Delta$ is known to the nodes.
This is achieved by replacing the final phase of \alg{d2-color} (i.e., the last step) with a different approach.

In the final phase of the improved algorithm, the nodes cooperate to track the colors used by the d2-neighbors of each live node. Thus, they learn the \emph{remaining palette}: the set of colors of $[\Delta^2]$ not used by d2-neighbors. Gathering the information about a single live node is too much for a single node to accumulate, given the bandwidth limitation. Instead, each live node chooses a set of \emph{handlers}, each handling a subrange of its color spectrum. The colored nodes then need to forward their color to the appropriate handler of each live d2-neighbor. After learning about the colors used, each of the multiple handlers choose an unused color at random and forward it to the live node. The live node selects among the proposed colors at random and tries it (which works with constant probability).

Since no routing information is directly available, we need to be careful how the coloring information is gathered at the handlers. We use here a \emph{meet-in-the-middle} approach. Each handler informs a random subset of its d2-neighbors about itself and each colored node sends out its message along a host of short random walks. In most cases, if the numbers are chosen correctly, a random walk will find an informed node, which gets the message to the handler. 

Once the unused palette is available, the coloring can be finished up in $O(\log n)$ rounds in the same fashion as the basic randomized {\congest} algorithm for ordinary coloring.

\begin{quote}
   \textbf{Algorithm} \emph{Improved-d2-Color}

   If $\Delta^2 \ge c_2\log n$ then  \\
   \hspace*{2em} repeat $c_0 \log n$ times: \\
   \hspace*{4em} Each live node picks a random color and \emph{tries} it. \\     
   \hspace*{2em} Form the similarity graphs $H=H_{2/3}$ and $\hat{H} = H_{5/6}$ \\
   \hspace*{2em} for ($\tau \leftarrow c_1 \Delta^2$; $\tau > c_2\log n$; $\tau \leftarrow \tau/2$) \\
   \hspace*{4em}      \alg{Reduce}($2 \tau$, $\tau$) \\
    \alg{LearnPalette}() \\
    \alg{FinishColoring()}
\end{quote}

The main effort of this section is showing how to learn the remaining palette in $O(\log n)$ steps.
We first show how that information makes it easy to color the remaining nodes.

\mypar{Finishing the coloring}
Suppose each node has $O(\log n)$ live d2-neighbors and knows the remaining palette. This includes the case when $\Delta^2 \le c_2\log n$, in which case no d2-neighbors are yet colored. 
We can then simulate the basic randomized algorithm for ordinary colorings with constant overhead, to complete the coloring in $O(\log n)$ rounds.

This algorithm, \alg{FinishColoring}, proceeds as follows:
Each node $v$ repeats the following two-round procedure until it is successfully colored. 
Flipping a random coin, $v$ is quiet or tries a random color from $T'_v$, with equal probability $1/2$.
If it succeeds, it forwards that information to immediate neighbors. They promptly forward it to each of their immediate neighbors $w$, who promptly updated their remaining palette $T'_w$. If a node has a backlog of color notifications to forward, it sends out a \textsc{Busy} message. A node with a \textsc{Busy} neighbor then waits (stays quiet) until all notifications have been forwarded (and all \textsc{Busy} signals have been lifted from its immediate neighbors). 

\begin{lemma}
\alg{FinishColoring} completes in $O(\log n)$ rounds, w.h.p.
\label{l:step7}
\end{lemma}

\begin{proof}
A node waits for a busy neighbor for at most $O(\log n)$ rounds, since it has that many live d2-neighbors.
Consider then a non-waiting round. Since it is not waiting, it knows its palette exactly.
With probability 1/2, at least half of the live d2-neighbors of $v$ are quiet. In this case, at least half of the colors of $v$'s palette are not tried by d2-neighbors, and hence, $v$ succeeds with probability at least $1/2$. The expected number of non-waiting rounds is therefore $O(\log n)$, and by Chernoff (\ref{eq:concentration}), this holds also w.h.p.
\end{proof}

\mypar{Learning the Available Palette}
Let $\varphi \le c\log n$ be an upper-bound on the leeway of live nodes.
Let $Z$ and $P$ be quantities to be determined.
We call each set $B_i = \{i\cdot \Delta^2/Z, i\cdot \Delta^2/Z+1, i\cdot \Delta^2/Z+2, \ldots, (i+1)\Delta^2/Z-1 \}$, $i=0,1,\ldots, Z-1$, a \emph{block} of colors. The last block $B_{Z-1}$ additionally contains the last color, $\Delta^2$. 
There are then $Z$ blocks 
that partition the whole color space $[\Delta^2]$.

\medskip

   \textbf{Algorithm} \emph{LearnPalette}()
   
    \emph{Precondition}: Live nodes have leeway at most $\varphi \le c_2\log n$.

  \emph{Postcondition}: Live nodes know their remaining palette

\begin{enumerate}
  \item If $\Delta = O(\log n)$, then the nodes learn the remaining palette in $\Delta$ rounds by flooding, and halt.
  \item Each node learns of its live d2-neighbors by flooding.
  \item For each live node $v$ and each block $i \in \{1,\ldots, \Delta\}$ of colors, a random $H$-neighbor $z^i_v$ of $v$ is chosen.
  \item Each node $z^i_v$ picks a random subset $Z_v^i$ of $P$ d2-neighbors (formed as a set of random 2-hop paths).
  It informs them that it "handles" block $i$ of the palette of live node $v$ (which indirectly tells them also the 2-path back to $z^i_v$). 
  \item Each colored node $u$ with color $c_u$ attempts to forward its color to some node in $Z_v^i$, where $i = \lfloor c(u) / \Delta \rfloor$, for each live d2-neighbor $v$. 
  This is done by sending the color along $\Theta(\Delta^2/P \cdot \log n)$ different random 2-paths. The node in $Z_v^i$ then forwards it directly to $z^i_v$. Let $C^i_v$ denote the set of colors that $z^i_v$ learns of.
  \item Each node $z_v^i$ informs $v$ by pipelining of the set $T_v^i = B_i \setminus C_v^i$ of colors missing within its range. 
  \item $v$ informs its immediate neighbors by pipelining of $T_v= \cup_i T_v^i$, the colors that it has not learned of being in its neighborhood. Each such node $w$ returns the set $\hat{T}_{v,w}$, consisting of the colors in $T_v$ used among $w$'s immediate neighbors. $v$ removes those colors from $T_v$ to produce $T'_v = T_v \setminus \cup_w \hat{T}_{v,w}$, which yields the true remaining palette $[\Delta^2] \setminus T'_{v}$.
\end{enumerate}
 \medskip
 
We first detail how a node $u$ selects a set of $m$ random d2-neighbors, as done in Steps 3, 4, and 5. It picks $m$ edges (with replacement) to its immediate neighbors at random and informs each node $w$ of the number $m_w$ of paths it is involved in. Each immediate neighbor $w$ then picks $m_w$ immediate neighbors. This way, $u$ does not directly learn the identity of the d2-neighbors it selects, but knows how to forward messages to each of them. Broadcasting or converge-casting individual messages then takes time $\max_{w \in N_G(v)} m_w$, which is $O(m/\Delta + \log n)$, w.h.p. (by (\ref{eq:concentration})).

The key property of this phase is the following.

\begin{lemma}
$|T_v| = O(\log n)$, for every live node $v$, w.h.p.
\label{l:last}
\end{lemma}

\begin{proof}
By assumption, live node $v$ has leeway $O(\log n)$ at the start of the algorithm, and thus it has slack $O(\log n)$. 
By the contrapositive of Prop.~\ref{P:sparsity}, it is $\zeta$-sparse, for $\zeta=O(\log n)$. 
Thus, the $H$-degree of $v$ is at least $\Delta^2 - 40\zeta$, by Lemma~\ref{L:h-degree}(1). For all $H$-neighbors of $v$, a random 2-hop walk has probability at least $|Z_v^i|/\Delta^2 = P/\Delta^2$ of landing in $Z_v^i$. 
Thus, w.h.p., one of the $\Theta((\Delta^2/P)\log n)$ random walks ends there, resulting in the color being recorded in $C_v$. 
Hence, w.h.p., $|T_v| \le |N_{G^2}(v)\setminus N_H(v)| \le 40\zeta = O(\log n)$.
\end{proof}

A careful accounting of the time spent yields that the dominant terms of the complexity are: $O(\log n)$ (Steps 2, 6-7), $O(PZ\varphi/\Delta^3) = O(PZ (\log n)/\Delta^3)$ (Step 4), $O(\Delta\varphi/P \log n) = O(\Delta(\log n)^2/P)$ (first half of Step 5), and $O(\Delta/Z \log n)$ (second half of Step 5). To optimize, we set $Z = \Delta$ and $P = \Delta \sqrt{\Delta\log n}$, for time complexity of $O(\log n (1 + \sqrt{(\log n)/\Delta}))$, which is $O(\log n)$ when $\Delta = \Omega(\log n)$. 

\begin{theorem}
The time complexity of \alg{LearnPalette}($\varphi$) with $\varphi=O(\log n)$ is $O(\log n)$, when $\Delta = \Omega(\log n)$.
\label{T:learnpalette}
\end{theorem}

Combining Thm.~\ref{T:learnpalette} and Lemma \ref{l:step7} with
Thm.~\ref{T:reduce-correct} of the previous subsection,
we obtain our main result.

\smallskip
\textsc{\Cref{thm:d2ColoringRand}.} 
\emph{There is a randomized \CONGEST algorithm that d2-colors a graph with $\Delta^2+1$ colors in $O(\log \Delta \log n)$ rounds, with high probability.}

\section{Deterministic $G^2$-Coloring}
\label{S:det-diam}
Completely independent of the rest of this section we use \Cref{ssec:sumDeltaDeta} to summarize our results on efficient deterministic algorithms when the dependence on $n$ is limited, i.e., algorithms with a runtime of $O(f(\Delta)+\logstar n)$. 
The goal of the current section is to color the square $G^2$ of the network graph $G$ with $(1+\eps)\Delta^2$ colors in polylogarithmic time for some $\eps>0$ and some globally known upper bound $\Delta$ on the maximum degree of the graph $G$. 

\paragraph{Coloring $G$ with $(1+\eps)\Delta$ colors in $\polylog n$ rounds:} If we were to compute a $(1+\eps)\Delta$ coloring of $G$ instead of $G^2$, we could recursively split $G$ into two graphs with roughly half the maximum degree. If we could halve the maximum degree precisely enough such that after $h=O(\log \Delta)$ recursion levels, we would have $p_h=2^h$ graphs each with maximum degree $\Delta_h=(1+\eps)2^{-h}\Delta$. We could then simply color each of them in $O(\Delta_h+\logstar n)$ rounds with a distinct color palette with $\Delta_h+1$ colors each (e.g. using the algorithm in \cite{BEG18}) and obtain a $p_h\cdot (\Delta_h+1)\approx(1+\eps)\Delta$ coloring of $G$. In \Cref{ssec:derandSplitting} we show that one can indeed deterministically split the original network graph with the necessary precision efficiently in the {\congest} model, and 
we use that in this section to color $G$.  We obtain the deterministic splitting algorithm from derandomizing a simple randomized algorithm with the method of conditional expectation. For more details on the derandomization we refer to \Cref{ssec:derandSplitting}; here, we only want to point out that most care is needed when formally reasoning that vertices can compute certain conditional expectations. When computing splittings of the original network graph $G$,
this only depends on information that $v$ can easily learn, e.g., because it is contained in $v$'s immediate neighborhood.

\paragraph{Coloring $G^2$ with $(1+\eps)\Delta^2$ colors in $\polylog n$ rounds:} To color $G^2$ instead of $G$ we would like  to mimic the same approach. However, now the respective conditional expectations depends on information in the $2$-hop neighborhood of a node and in most graphs vertices cannot learn this information efficiently. Instead we proceed as follows: We split $G$ into $p=2^h$ graphs $G_1,G_2,\ldots,G_p$ of maximum degree $\Delta_h=(1+\eps/4)2^{-h}\Delta$. Then, for each $i=1,\ldots, p$ we consider the subgraph $H_i$ of $G^2$ that is induced by the vertices $V(G_i)$ of $G_i$. As $G_i$ has maximum degree $\Delta_h$ we obtain that $H_i$ has maximum degree $\Delta\cdot \Delta_h$ and furthermore we show that the graphs $H_1,\ldots, H_p$ are such that any {\congest} algorithm on each of the subgraphs can be executed in $G$ in parallel with a multiplicative $\Delta_h$ overhead in the runtime---this step needs additional care and a more involved definition of the splitting problem that we call \emph{local refinement splitting}. 
Then we use this property to apply the algorithm mentioned in the previous paragraph to color each $H_i$ in parallel with $(1+\eps/4)\Delta\cdot\Delta_h$ colors using a distinct color palette. The induced coloring is a coloring of $G^2$ with $(1+\eps)\Delta^2$ colors.

\smallskip
We now formally define the splitting problem that we need to solve. 
\begin{restatable}[Local Refinement Splitting]{definition}{refinementSplitting}
Let $G=(V,E)$ be a graph whose vertices are partitioned into $p\leq n$ groups $V_1,\ldots, V_p$, let $\lambda>0$ be a parameter and for $v\in V$ let $deg_i(v)$ be the number of neighbors of $v$ in $V_i$. An (improper) $2$-coloring of the vertices of $G$ with two colors (red/blue) is called a $\lambda$-local refinement splitting if each vertex $v$ with $deg_i(v)\geq 12\log n/\lambda^2$ has at most $(1+\lambda) \deg_i(v)/2$ neighbors of each color in set $V_i$ for all $i=1,\ldots,p$. 
\end{restatable}

One can show that the local refinement splitting problem is solved w.h.p. if each vertex picks one of the two colors uniformly at random and one can show that this holds even if nodes only use random coins with limited independence. We obtain the following result by derandomizing this zero-round algorithm with the help of a suitable network decomposition of $G^2$ which can computed efficiently with the results in \cite{RG19}.  
\begin{restatable}[Deterministic Local Refinement Splitting]{theorem}{derandRefinementSplitting}
\label{lem:derandSplitting}
For any $\lambda>0$ there is a $O(\log^8 n)$ round deterministic {\congest} algorithm to compute a $\lambda$-local refinement splitting.
\end{restatable}
Due to its length the formal proof of \Cref{lem:derandSplitting} (including the aforementioned claim about the randomized algorithm with limited independence) is deferred to  \Cref{app:localRefinementSplitting}. 

We now show how to use \Cref{lem:derandSplitting} to recursively split the graph deterministically into graphs with smaller maximum degree and further any vertex has a small number of neighbors in each such subgraph. Then, in \Cref{thm:coloringWithSplitting}, we use the former property of this partitioning result to compute a $(1+\eps)\Delta$ coloring of $G$ and in \Cref{thm:G2withSplitting} we use both properties to compute a $(1+\eps)\Delta^2$ coloring of $G^2$. 
\begin{lemma}
\label{lem:multiPhaseSplitting}
For any $\eps>0$ there is a $O(\log^8 n)$-round deterministic {\congest} algorithm to partition a graph into $p=2^h$ parts $V_1,\ldots, V_p$ such that every vertex $v\in V$ has at most 
$\Delta_h=(1+\eps)2^{-h}\Delta=O(\eps^{-2}\log^3 n)$
 neighbors in each $V_i$ where $h$ is the smallest integer such that $(1+\eps/(10\log \Delta))^h2^{-h}\Delta\leq 1200\eps^{-2} \log^3 n$ holds.
\end{lemma}
\begin{proof}
Let $\eps'=\min\{1,\eps/4\}$ and let $h=O(\log \Delta)$ be the smallest integer such that \begin{align}
(1+\eps/(10\log \Delta))^h2^{-h}\Delta\leq 1200\eps^{-2} \log^3 n~.
\end{align}

Then recursively apply \Cref{lem:derandSplitting} with $\lambda=\eps'/(10\log \Delta))$ where we begin with the trivial partition $V_1=V$ and with each recursion level each part of the partition is naturally split into two parts according to the two colors of the local refinement splitting. The \emph{output partition} of recursion level $i$ serves as the \emph{input partition} for recursion level $i+1$. 
 Due to the choice of $h$ we have that the guaranteed maximum degree $\Delta_i=(1+\lambda)^h2^{-h}\Delta$ after recursion level $i< h$ is at least at least $12\log n /\lambda^{2}$ and thus it decreases by a $(1+\lambda)2^{-1}$ factor with each iteration, that is, after iteration $h$ we have $p=2^h$ subgraphs $G_1,\ldots,G_p$, each with maximum degree at most  
\begin{align}
(1+\eps/(10\log \Delta))^h2^{-h}\Delta & \stackrel{(*)}{\leq} e^{\eps/(10\log \Delta)\cdot h}2^{-h}\Delta\leq e^{\eps/10}2^{-h}\Delta\\
& \stackrel{(**)}{\leq} (1+\eps/5)2^{-h}\Delta\leq (1+\eps)2^{-h}\Delta=\Delta_h~.
\end{align}
At $(*)$ we used that $(1+x)\leq e^x$ and $(**)$ we used $e^x\leq 1+2x$ for $0\leq x\leq 1$. 
Further, with the same calculation---recall that we solve a local refinement splitting in each recursion level---any vertex has at most $(1+\eps)2^{-h}\Delta$ neighbors in each $V_i$. 
Further, by the definition of $h$, we obtain $\Delta_h=O(\eps^{-2}\log^3 n)$~.

We obtain a runtime of $h\cdot O(\log^8 n)=O(\log^9 n)$ rounds  by $h$ applications of \Cref{lem:derandSplitting}. The runtime in the proof of \Cref{lem:derandSplitting} is dominated by the computation of a so called \emph{network decomposition}. If we reuse the same network decomposition in each call of \Cref{lem:derandSplitting} the runtime can be reduced to $O(\log^8 n)$ rounds. 
\end{proof}
We now use \Cref{lem:multiPhaseSplitting} to compute a $(1+\eps)$-coloring of $G$. 
\begin{theorem}
\label{thm:coloringWithSplitting}
For any constant $\eps>0$ there is a deterministic {\congest} algorithm that computes a $(1+\eps)\Delta$ coloring of $G$ in $O(\log^8 n+\eps^{-2}\log^3 n)$ rounds.
\end{theorem}
\begin{proof}
If $\Delta=O(\eps^{-2}\log^3 n)$ use, e.g., the algorithm of \cite{BEG18} to color the graph with the desired number of colors in $O(\eps^{-2}\log^3 n+\logstar n)$ rounds.
Otherwise, apply \Cref{lem:multiPhaseSplitting} with $\eps'=\eps/2$ to obtain a partition $V_1,\ldots,V_p$ of the vertices where $p=2^h$ (where $h$ is chosen as in \Cref{lem:multiPhaseSplitting}) and the maximum degree of $G[V_i]$ is at most $\Delta_h=(1+\eps/2)2^{-h}\Delta$. 
Then color each of the subgraphs in parallel (no vertex nor an edge is used in more than one subgraph) with a distinct set of $\Delta_h+1$ colors in $O(\Delta_h+\logstar n)=O(\eps^{-2}\log^3 n)$ rounds using e.g. the algorithm of \cite{BEG18}. The induced coloring is a proper coloring of $G$ as we use disjoint color palettes for distinct subgraphs and the total number of colors is
\begin{align}
2^h\cdot (\Delta_h+1) & =2^h\cdot (1+\eps/2)2^{-h}\Delta+2^h\leq (1+\eps/2)\Delta+2^h\leq  (1+\eps)\Delta~.
\end{align}
In the very last inequality we used that $h$ is the smallest integer with the aforementioned property from which one can deduce that
$2^h\leq \eps/2\cdot \Delta$.
As a runtime we obtain $O(\log^8 n)$ rounds from the application of \Cref{lem:multiPhaseSplitting} and $O(\eps^{-2}\log^3 n)$ rounds from coloring the subgraphs or the whole graph if the maximum degree $\Delta$ is in $O(\eps^{-2}\log^3 n)$ to begin with.
\end{proof}

To color $G^2$ we first show that a partition obtained by recursively applying a local refinement splitting is helpful to run {\congest} algorithms on the induced  subgraphs of $G^2$ in parallel. 
\begin{lemma} 
\label{lem:CONGESTsimulation}
Let $V_1,\ldots,V_p$ be a partition of the graph such that every vertex $v\in V$ has at most $\Delta'$ $G$-neighbors in $V_i$ for each $i=1,\ldots, p$ and let $A_1, \ldots, A_p$ be algorithms where algorithm $A_i$ runs on $H_i=G^2[V_i]$. 
Then, we can execute one round of each of the algorithms in parallel in $O(\Delta')$ {\congest} rounds in $G$. 
\end{lemma}
\begin{proof}
We first let all vertices in all $H_i$ send messages to their neighbors that are also neighbors in $G$. 
Now, if a vertex $v$ needs to send a message to a neighbor $u$ in $H_i$ that is not an immediate neighbor in $G$, that is, $u$ and $v$ are only connected in $G$ through a node $w$, then $v$ sends the message to $w$ and $w$ forwards it to $u$. By the construction of $H_i$ for each $i$ each vertex $v$ of $G$ has at most $\Delta_h$ neighbors that are vertices of $H_i$. Each of these $\Delta_h$ neighbors can have at most $1$ message for any vertex in $N_G(v)\cap V_i$. Thus $v$ has to forward at most $\Delta_h$ messages to a single vertex which it can do in $O(\Delta_h)$ rounds by pipelining. 
\end{proof}
With the {\congest} simulation result from \Cref{lem:CONGESTsimulation} we can run \Cref{thm:coloringWithSplitting} on the parts of a partition that we computed \Cref{lem:multiPhaseSplitting} to obtain the main result of this section.
\smallskip

\textsc{\Cref{thm:G2withSplitting} (Full)}. 
\emph{For any  $\eps>0$ there is a deterministic {\congest} algorithm that computes a $(1+\eps)\Delta^2$ coloring of $G^2$ in $O(\eps^{-2}\log^{11} n+\eps^{-4}\log n)$ rounds. For constant $\eps>0$ the runtime is $O(\log^{11}n)$.}

\begin{proof}
Let $\eps'=\min\{1,\eps/4\}$ and apply \Cref{lem:multiPhaseSplitting} with $\eps'$ to obtain  $p=2^h$ subgraphs $G_1,\ldots,G_p$, each with maximum degree at most ($h$ is chosen as in \Cref{lem:multiPhaseSplitting})
\begin{align}
\Delta_h=O(\eps^{-2}\log^3 n)~.
\end{align}
Then, for $i=1,\ldots,p$ let $H_i$ be the subgraph of $G^2$ that is induced by the vertices $V(G_i)$ of $G_i$. Here each vertex knows which subgraph it belongs to and also which of its immediate neighbors in $G$ belong to which subgraph.
 As $G_i$ has maximum degree $\Delta_h$ we obtain that $H_i$ has maximum degree $\Delta\cdot \Delta_h$. 

Further any vertex has at most $\Delta_h$ neighbors in any $H_i$ and 
with \Cref{lem:CONGESTsimulation} we can execute the algorithm from \Cref{thm:coloringWithSplitting} with $\eps'$ in parallel on each subgraph $H_1, \ldots, H_p$ to color each of them with a distinct set of colors with size $(1+\eps')\Delta_h\cdot \Delta$ in time $O(\Delta_h\cdot (\log^8 n+\eps^{-2}\log^3 n))=O(\eps^{-2}\log^{11} n+\eps^{-4}\log^3 n)$.
The computed coloring forms a coloring of $G^2$ with $2^h\cdot (1+\eps')\Delta_h\cdot \Delta=(1+\eps')\cdot (1+\eps')\Delta^2\leq(1+\eps)\Delta$ colors. 
\end{proof}

\begin{remark}
We emphasize that we did not attempt to reduce the $\log n$ factors or to express as many of them as $\log \Delta$ factors as possible. 
\end{remark}


\subsection{Summarizing the Ideas for \Cref{thm:d2ColoringDelta} }
\label{ssec:sumDeltaDeta}
In this section we summarize our algorithm to obtain efficient algorithms for coloring $G^2$ with $\Delta^2+1$ colors if $\Delta$ is small, that is, we explain the core ideas of the following theorem.

\smallskip
\textsc{\Cref{thm:d2ColoringDelta}.} 
\emph{There is a deterministic \CONGEST algorithm that d2-colors a graph with $\Delta^2+1$ colors in $O(\Delta^2+\logstar n)$ rounds.}
\smallskip

The algorithm of \Cref{thm:d2ColoringDelta} has three components that are executed in the presented order:

\smallskip

\noindent\textbf{ Linial: $O(\Delta^4)$ Colors (\Cref{thm:Liniald2}):}  A standard pipelined version of Linial's algorithm run on $G^2$ computes a $O(\Delta^4)$-coloring of $G^2$ in $O(\Delta\cdot \logstar n)$ rounds. We show that the runtime can be reduced to $O(\Delta+\logstar n)$ rounds. 

\smallskip

\noindent\textbf{Locally Iterative: $O(\Delta^4)\rightarrow O(\Delta^2)$  (\Cref{thm:d2locIt}) } In \cite{barenboim15,BEG18} it was shown that there is a {\congest} algorithm that colors the network graph $G$ with $O(\Delta)$ colors in $O(\sqrt{\Delta})$ rounds given an $O(\Delta^2)$-coloring of the input graph. With \Cref{thm:Liniald2} we can compute a $O((\Delta(G^2))^2)=O(\Delta^4)$-coloring of $G^2$ in $O(\Delta+\logstar n)$ rounds. Using this coloring we run the algorithm from \cite{barenboim15,BEG18} on $G^2$ and simulate one round of it in $\Delta$ rounds of communication on $G$. That way, we obtain a coloring of $G^2$ with $O(\Delta(G^2)=O(\Delta^2)$ colors and the runtime is $O(\Delta\cdot \sqrt{\Delta(G^2)})=O(\Delta^2)$.  
As the combination of \cite{barenboim15,BEG18} is slightly involved (e.g., it includes the computation of arbdefective colorings) we present a self contained algorithm for coloring $G^2$ with $O(\Delta^2)$ colors in $O(\Delta^2+\logstar n)$ rounds. Our algorithm is based on the \emph{locally iterative} algorithm in \cite{BEG18}. 

\smallskip

\noindent\textbf{Color Reduction: $O(\Delta^2)\rightarrow \Delta^2+1$ (\Cref{thm:delta2_simple_color}).} In the \emph{iterative color reduction} for $G$ one iteratively let's nodes with the largest color class pick a smaller color until one obtains a coloring with $\Delta(G)+1$ colors. The crux in implementing this algorithm for $G^2$ is, that nodes, need to know all colors that are used in its $d2$-neighborhood to recolor themselves. A naive approach to learn these colors would take $\Delta$ rounds for each recoloring step and result in a runtime of $O(\Delta^3)$ rounds. Using the fact, that at most one vertex in each neighborhood of a node changes its color in one round we show that the simple color reduction can be done in $O(\Delta^2)$ rounds.

\bibliographystyle{abbrv}
\bibliography{refs}

\clearpage
\appendix

\section{Derandomization of Local Refinement Splitting}
\label{app:localRefinementSplitting}
We begin with the definition of a network decomposition that is adapted to the {\congest} model and to power graphs. Then, in \Cref{ssec:derandSplitting} we use a network decomposition of $G^2$ (that can be computed with the algorithm from \cite{RG19}) to derandomize a zero-round algorithm for local refinement splitting.

\label{ssec:networkDecomp}
\begin{definition}[Network decomposition with congestion, \cite{RG19}]
\label{def:nd} Let $k>0$ be an integer. 
An $(\alpha,\beta)$-network decomposition with congestion $\kappa$ of $G^k$ of a graph $G=(V,E)$ is a partition of $V$ into clusters $C_1,\dots,C_p$ together with associated subtrees $T_1,\ldots,T_p$ of $G$ and a color $\gamma_i\in\{1,\dots,\alpha\}$ for each cluster $C_i$ such that
\begin{enumerate}[(i)]
\item the tree $T_i$ of cluster $C_i$ contains all nodes of $C_i$ (but it might contain other nodes as well)
\item each tree $T_i$ has diameter at most $\beta$
\item clusters that are connected by a path of length $\leq k$ in $G$ are assigned different colors
\item each edge of $G$ is contained in at most $\kappa$ trees of the same color
\end{enumerate}
\end{definition}

When we assume to have a network decomposition on a graph, we require that each node knows the color of the cluster it belongs to and for each of its incident edges $e$ the set of associated trees $e$ is contained in. Note that a decomposition according to this definition has weak diameter $\beta$ and a strong network decomposition is a decomposition with congestion 1 where the tree $T_i$ of each cluster $C_i$ contains exactly the nodes in $C_i$.

\begin{theorem}[\cite{RG19}]
\label{thm:rozhon}
There is a deterministic algorithm that computes an $\left(O(\log n),k\cdot O(\log^3n)\right)$-network decomposition of $G^k$ with congestion $O(\log n)$ in $O(k\cdot \log^8 n)$ rounds in the {\congest} model.\footnote{In on-going unpublished work \cite{RG19personal} it is shown that diameter and runtime can be improved. These improvements carry over to our results.}
\end{theorem}

 \label{ssec:derandSplitting}

As long as degrees are at least $\polylog n$ the degrees of all nodes can be split roughly in half without communication by letting each vertex choose one of two subgraphs uniformly at random. We show that this algorithm can be derandomized with a network decomposition of $G^2$. 
Formally, we  solve the following more general version of the problem efficiently and deterministically.

\refinementSplitting*

Note that we only require a bound on vertices with a degree of at least $12\log n/\lambda^2$. Thus coloring each vertex red or blue with probability $1/2$ each solves the problem with high probability. More formally, introduce a random variable $F^i_v\in \{0,1\}$ for each node $v\in V$ and each $i=1,\ldots,p$, where $F^i_v=1$ if $\deg_i(v)\geq 12\log n/\lambda^2$ and  $v$  has more than $(1+\lambda) \deg_i(v)/2$ neighbors of one color in $V_i$, and $F^i_v=0$ otherwise. The \emph{flag} $F^i_v$ indicates whether the splitting failed locally for $v$ in set $V_i$. 
By a Chernoff  bound one can show that  $Pr(F^i_v=1)\leq 1/n^2$ for all vertices and all $i$. Let $F_v=\sum_{1\leq i\leq p}F^i_v$. By a union bound over the $p< n$ parts we obtain that $Pr(F_v=1)<1/n$  and thus we obtain 
\begin{align}
E[\sum_{v\in V}F_v]<1~.
\end{align}
This will later be sufficient to derandomize the algorithm. To make the derandomization efficient
 we next show that we can similarly bound the sum of the expectations of $F_v, v\in V$ if the random choices of the vertices are only $\Theta(\log n)$-wise independent. First we define the notion of limited independence and restate a Chernoff bound that holds with limited independence. 

\begin{definition}[\cite{Vadhan12}]
For $N,M,k\in\mathbb{N}$ such that $k\leq N$, a family of functions $\mathcal{H}=\{h:[N]\to [M]\}$ is $k$-wise independent if for all distinct $x_1,\dots,x_k\in[N]$, the random variables $h(x_1),\dots,h(x_k)$ are independent and uniformly distributed in $[M]$ when $h$ is chosen uniformly at random from $\mathcal{H}$.
\end{definition}
We use the following Chernoff bound that works with limited independence.
\begin{theorem}[Theorem 5 in \cite{chernoffLimited95}] \label{thm:kindependent}
Let $X$ be the sum of $k$-wise independent $[0,1]$-valued random variables with expectation $\mu=E(X)$ and let $\delta\leq 1$. Then we have
\[Pr(|X-\mu|\geq \delta\mu)\leq e^{-\lfloor \min\{k/2,\delta^2\mu/3\}\rfloor}~.\] 
\end{theorem}

The next lemma states $\polylog n$-wise independent random bits for the vertices are sufficient to obtain the aforementioned bound on the expected sum of the flags. 
\begin{lemma}
\label{lem:limIndProcess}
Let $\lambda>0$ be a parameter, let $G=(V,E)$ be a $n$-node graph, with vertex partition $V_1,\ldots,V_p$ and 
let each vertex of $G$ color itself red or blue with probability $1/2$ each where random choices of distinct nodes are $10\log n$-wise independent. 
Then we have $E[\sum_{v\in V}F_v]<1$. 
\end{lemma}

\begin{proof}
For $1\leq i\leq p$ and a vertex $v$ with  $deg_i(v)\geq 12\log n/\lambda^2$ let $X_i$ be the number of red neighbors in $V_i$. We obtain $\mu_i=E[X]=\deg_i(v)/2\geq 6\log n /\lambda^2$ and $X_i$ is the sum of $[0,1]$-valued $k$-wise independent random variables with $k=10\log n$. 
By \Cref{thm:kindependent} with $\delta=\lambda$ we obtain that the probability that $v$ has more or less than $(1\pm\lambda)\deg_i(v)/2$ red neighbors in $V_i$ is bounded by
\begin{align}
Pr(|X_i-\mu_i|\geq \lambda\mu_i)\leq e^{-\lfloor \min\{k/2,\lambda^2\mu_i/3\}\rfloor} \leq e^{-\lfloor \min\{5\log n,12\log n/3\}\rfloor}< e\cdot n^{-4}~.
\end{align}
The exact same analysis also holds for the number of blue neighbors in $V_i$ and with a union bound over both colors and all $1\leq i\leq p< n$ we obtain that $Pr(F_v=1)< 1/n$. 
Because $F_v$ is a $0$ or $1$ valued random variable this implies 
\begin{align*}
    E\left[\sum_{v\in V}F_v\right]<1~.  & & \qedhere
\end{align*}
\end{proof}

To derandomize the aforementioned algorithm we need to produce random coins for the vertices with limited independence from short random seeds as given by the next theorem.

\begin{theorem}[\cite{Vadhan12}]\label{thm:seed}
For every $a,c,k$, there is a family of $k$-wise independent hash functions $\mathcal{H}=\{h:\{0,1\}^a\to\{0,1\}^c\}$ such that choosing a random function
from $\mathcal{H}$ takes $k\cdot\max\{a,c\}$ random bits.
\end{theorem}
In the proof of our main derandomization result (\Cref{lem:derandSplitting}) we will use one random seed for each cluster of a suitable network decomposition.  
In particular, we can use \Cref{thm:seed} to produce fair random coins for up to  $n$ vertices (of a cluster) that are $k$-wise independent for $k=\Theta(\log n)$ from a random seed of length $k\cdot \max\{a,c\}=O(\log^2 n)$ by setting $a=O(\log n)$ and $c=1$. 
Then we obtain a deterministic algorithm for computing a local refinement splitting by iterating through the color classes of the decomposition and  using  the method of conditional expectation with objective function $E[\sum_{v\in V}F_v]$ to find a \emph{good seed} for each cluster. Afterwards the 'randomized' algorithm is executed with the good seeds and we obtain $E[\sum_{v\in V}F_v]<1$, i.e., $E[\sum_{v\in V}F_v]=0$ and the flag $F_v$ of every vertex $v\in V$ equals zero, that is, we have computed a local refinement splitting. Formally, we show the following result.

\derandRefinementSplitting*

\begin{proof}[Proof of \Cref{lem:derandSplitting}]
Let $V_1,\ldots,V_p$ be the given partition of the vertices. 
We first use \Cref{thm:rozhon} to compute a $(O(\log n),O(\log^3 n))$-network decomposition of $G^2$ with congestion $O(\log n)$ in $O(\log^8 n)$ rounds; note that this network decomposition ignores the partition $V_1,\ldots,V_p$. Then, with a random seed of length $\Theta(\log^2 n)$ for each cluster (apply \Cref{thm:seed} with $k=\Theta(\log n)$, $a=O(\log n)$ and $c=1$) we can produce fair coins, one for each node of the cluster,  that are $10\log n$-wise independent within the cluster and completely independent for distinct clusters. Here, we use the same globally known procedure in each cluster to compute the coin of a vertex $v$ given its ID and the outcome of the random seed. 

Recall that $F_v$ denotes the random variable that indicates whether the splitting failed locally at a  vertex $v$ for any $V_i$, $1\leq i\leq p$.
By \Cref{lem:limIndProcess}, we have $E[\sum_{v\in V}F_v]<1$ over the randomness of all random seeds.  The objective is to find a collection of \emph{good seeds} 
such that the 'randomized' algorithm if deterministically executed with the good seeds satisfies $\sum_{v\in V}F_v<1$. 
To this end we iterate through the color classes of the network decomposition and in parallel fix the seeds of clusters of the same color class. These seeds can be fixed independently and in parallel as any two clusters with the same color are at least three hops apart and thus none of the random variables $F_v,v\in V$ is influenced by more than one cluster of the same color. We explain how to fix the random seed of one cluster $\mathcal{C}$. 

\paragraph{Deterministically finding a good random seed of cluster $\mathcal{C}$:}Let $N(\mathcal{C})$ be the set of vertices of $G$ that are contained in the cluster or have a neighbor in the cluster. Note that $N(\mathcal{C})$ has weak diameter $O(\log^3 n)+O(1)=O(\log^3 n)$.  We use the method of conditional expectation to iteratively determine the bits of the random seed $S^{\mathcal{C}}=S_1,\ldots,S_{l}$ of the cluster where $l=O(\log^2 n)$ is the seed length. 
Here for an already processed cluster $\mathcal{C}'$ we denote the already fixed seed by $s^{\mathcal{C}'}$.
To fix one bit $S_i$ of the seed $S^{\mathcal{C}}$ assume that the bits $S_j=s_j$, $j<i$ are already determined. We set bit $S_i\in\{0,1\}$ of the seed $S^{\mathcal{C}}$ as follows (we will later argue how a cluster leader can compute $s_i$) 
\begin{align}
s_i & =\argmin_{b\in \{0,1\}} \left\{E\left[\sum_{v\in V} F_v \mid S_i=b \wedge  \bigwedge_{j<i} S_j=s_j \wedge \bigwedge_{\text{processed cluster }\mathcal{C}'}S^{\mathcal{C}'}=s^{\mathcal{C}'} \right]\right\} \label{eqn:argmin}\\
& =\argmin_{b\in \{0,1\}} \left\{\sum_{v\in V}E\left[F_v \mid S_i=b \wedge  \bigwedge_{j<i} S_j=s_j \wedge \bigwedge_{\text{processed cluster }\mathcal{C}'}S^{\mathcal{C}'}=s^{\mathcal{C}'}  \right]\right\}~. \\
& \stackrel{(*)}{=}\argmin_{b\in \{0,1\}} \left\{\sum_{v\in N(\mathcal{C})}E\left[F_v \mid S_i=b \wedge  \bigwedge_{j<i} S_j=s_j \wedge \bigwedge_{\text{processed cluster }\mathcal{C}'}S^{\mathcal{C}'}=s^{\mathcal{C}'}  \right]\right\}
\end{align}

Equality $(*)$ follows as $F_v$ for $v\notin N(\mathcal{C})$ does not depend on the choice of $S_i$. 
First, note, that the law of total probability guarantees that for $S_i=s_i$ we have 
\begin{align} 
E\left[\sum_{v\in V} F_v \mid S_i=s_i \wedge  \bigwedge_{j<i} S_j=s_j \wedge \bigwedge_{\text{processed cluster }\mathcal{C}'}S^{\mathcal{C}'}=s^{\mathcal{C}'}  \right] \\
\leq E\left[\sum_{v\in V} F_v \mid  \bigwedge_{j<i} S_j=s_j \wedge \bigwedge_{\text{processed cluster }\mathcal{C}'}S^{\mathcal{C}'}=s^{\mathcal{C}'}  \right]
\end{align}
As we initially have $E[\sum_{v\in V}F_v]<1$, we obtain $\sum_{v\in V}F_v<1$ once all random seeds in all clusters are fixed. Thus the integer $\sum_{v\in V}F_v$ has to equal to $0$ if the algorithm is run deterministically with the fixed seeds of all clusters and in the process where each vertex uses its coin from its cluster's random seed to determine its color no vertex fails.

\paragraph{The value $s_i$ can be computed by a leader: }For $b\in \{0,1\}$ each vertex $v\in N(\mathcal{C})$ computes the two values 
$\alpha_b(v)=E[ F_v \mid S_i=b \wedge  \bigwedge_{j<i} S_j=s_j \wedge \bigwedge_{\text{processed cluster }\mathcal{C}'}S^{\mathcal{C}'}=s^{\mathcal{C}'} ]$~. To this end $v$ learns the IDs of its immediate neighbors (or the color of its immediate neighbor if some neighbors belong to a different cluster and 
already know their output color in the case that they are contained in a previously processed cluster) and the already fixed prefix of the random seed in their own cluster. 
We next formally prove that this information is sufficient to compute both values.

\medskip
\noindent\textbf{Claim:} With the IDs of its immediate neighbors, the colors of potentially already colored neighbors from previously processed clusters, the cluster ID of each of its adjacent vertices and the already fixed prefix in its own cluster $v$ can locally compute $\alpha_0(v)$ and $\alpha_1(v)$. 
\begin{proof}To determine the values $\alpha_0(v)$ and $\alpha_1(v)$ a node has to know the probability distribution of its neighbors being colored red or blue. This distribution is independent for neighbors in distinct clusters. Neighbors of already processed vertices are either colored red or blue and do not depend on any randomness. Neighbors of $v$ that are in the same cluster (this does not have to be the cluster of $v$) do not act independently but obtain their randomness from a shared random seed. As it is globally known how vertices use their ID to extract their decisions from the random seed and random seeds in all clusters are build identically, that is, each seed is long enough to produce coins for all IDs in the graph and it is globally known how to extract the value of the random coin related to a specific ID from a random seed, node $v$ can compute the desired two values with knowing the partition of its neighbors into different clusters.  For neighbors that are part of $v$'s cluster $\mathcal{C}$, $v$ additionally needs to know the prefix of the already fixed seed to determine the influence of the next bit of the seed on their probabilities.  
\end{proof}
After each vertex in $N(\mathcal{C})$ has computed its two values we aggregate two sums of these values in the cluster leader vertex. 
This aggregation can be done in $O(\log n\cdot \log^3 n)$ rounds over the tree corresponding to the cluster; here the additional $\log n$-term is due to the congestion in one color class. 
 The leader sets $s_i=0$ if $\sum_{v\in N(\mathcal{C})}\alpha_0(v)\leq \sum_{v\in N(\mathcal{C})}\alpha_1(v)$ and $s_i=1$ otherwise, then it distributes the choice to all vertices in $N(\mathcal{C})$. The random variables $F_v$ for all $v\notin N(\mathcal{C})$ do not depend on the choice of the seed and taking the better of the two aggregate values is identical to taking the desired $\argmin$ in \Cref{eqn:argmin}.

\paragraph{Runtime:}
Computing the network decomposition takes $O(\log^8 n)$ rounds. Afterwards we iterate through $O(\log n)$ color classes and the $O(\log^2 n)$ bits of the random seeds of each cluster of one color in parallel. Due to the congestion and the weak cluster diameter we need $O(\log^4)$ rounds to aggregate the values to decide on a single bit, that is, using the network decomposition has a complexity of $O(\log n\log^2 n\cdot \log^4 n)=O(\log^7 n)$. Note that a vertex can learn the necessary information (except for the seed prefix that it already knows by construction) to compute the values $\alpha_0(v)$ and $\alpha_1(v)$ in constant time.
\end{proof}


\section{Deterministic: $O(\Delta^2+\logstar n)$ rounds, $\Delta^2+1$ Colors (Theorem \ref{thm:d2ColoringDelta})}
\label{sec:g2-coloring}

We give a deterministic algorithm $O(\Delta^2+\logstar n)$-round algorithm for d2-coloring with $\Delta^2+1$ colors. In fact, the algorithm uses at most $\Delta(G^2)$ colors. Before presenting the algorithm, we describe how well-known pre- and post-processing steps can be implemented efficiently.

\subsection{Linial \& Color Reduction for d2-Coloring}
\paragraph{Refined Analysis of the Linial's Algorithm for $G^2$}

Linial's famous  $O(\logstar n)$ algorithm to compute a $O(\Delta^2)$-coloring of the communication graph $G$ can be implemented in the {\congest} model. In the d2-setting, the algorithm would imply a $O(\Delta^4)$-coloring. A naive implementation in {\congest} (with queuing messages) would implement one round of Linial in $G^2$ in $\Delta$ rounds on $G$ and would yield $O(\Delta\cdot \logstar n)$ rounds. 

We now show that the algorithm can be modified to run in $O(\Delta^2+\logstar n)$ rounds in {\congest}.
Without going into the details of the algorithm, the key point is that in round $i$ the nodes compute a valid coloring using $O(\log^{(i)} n)$ colors, the $i$-th repeated logarithm. 
\begin{theorem}[Linial's algorithm]
\label{thm:Liniald2}
There is a deterministic {\congest} algorithm that d2-colors a graph with $O(\Delta^4)$ colors in $O(\Delta + \log^* n)$  rounds.
\end{theorem}

\begin{proof}
We simulate the first two iterations of Linial's method in $2\Delta$ round, by pipelining.
This reduces the number of colors to $O(\max(\Delta^4, \log\log n))$. 
If $\Delta^4 \ge \log\log n$, then we are done.
So, we continue with the case $\Delta^4 \le \log\log n$.
In each iteration, nodes need to distribute to their neighbors the colors in their neighborhood.
These are the colors of their $\Delta$ neighbors, each an integer in the range $[1,\ldots, O(\log\log n)]$. This can be encoded in a bitstring of size $\Delta \cdot O(\log\log\log n) = o(\log n)$. Hence, this can be encoded in a single {\congest} message, when $n$ is large enough.
Hence, the remaining $\log^* n$ iterations of Linial require only $\log^* n$ rounds, for a total round complexity of $2\Delta + \log^* n$.
\end{proof}

\paragraph{Iterative Color Reduction  for $G^2$}
Given a $(c+k)$-coloring of $G$ with $c\geq \Delta+1$ we can compute a $c$-coloring of $G$ in $O(k)$ by iteratively recoloring the vertices of the largest color class with a color in $[c]$ that is not used by its already colored neighbors, that is, we use $O(1)$ rounds to remove one color class. 

\begin{theorem}
\label{thm:delta2_simple_color}
Given a $(c+k)$-coloring of $G^2$ with $c\geq \Delta(G^2)+1$ and assume that each color can be represented with $O(\log n)$ bits, then we can compute a $c$-coloring of $G^2$ in $O(\Delta+k)$ rounds in {\congest} on $G$. 

In particular, we can compute a $(\Delta(G^2)+1)$-coloring given a $O(\Delta^2)$-coloring in $O(\Delta^2)$ rounds.
\end{theorem}
\begin{proof}
First, in $\Delta$ rounds, each vertex learns all colors in its $d2$-neighborhood. During the algorithm we maintain the invariant that each vertex always knows all used colors in its $d2$-neighborhood. 
Using this invariant, in one phase, any vertex $v$ that has a color $\alpha>c$ that is larger than any other color in its two hop neighborhood picks a color that is not used by any of its $d2$-neighbors (there is always a free color in the color space of size $\Delta(G^2)+1$ as at most $\Delta(G^2)$ colors can be used by neighbors) and informs its neighbors of its new color pick by broadcasting its new color pick for two hops. All nodes satisfying the requirement perform this recoloring step in parallel. 
There is no congestion during this broadcasting because after the first hop each node can receive a message from at most one of its neighbors as no vertices of distance two can both take part in the recoloring step in the same phase. The algorithm is iterated for $k$ phases.
\end{proof}

\subsection{Locally Iterative d2-Coloring}

In the following we design a {\congest} algorithm to find a proper $\Delta^2+1$-coloring of $G^2$ in $O(\Delta^2+\logstar n)$ rounds where, as in all previous sections, the communication network is the graph $G$ itself.  Our algorithm is based on the algorithm in \cite{BEG18} and we obtain the following theorem.

\medskip
\noindent \textbf{\Cref{thm:d2ColoringDelta}.} 
There exists a deterministic CONGEST algorithm that d2-colors a graph with $\Delta^2+1$ colors in $O(\Delta^2+\logstar n)$ rounds.

\medskip

\paragraph{Locally Iterative Coloring Algorithm:} The algorithm consists of three steps. First we compute a $10\Delta^2$-coloring of $G^2$ using the adaption of Linial's algorithm from \Cref{thm:Liniald2}. 
Then, using its input color each node $v$ computes a sequence $p_v(0),\ldots, p_v(q-1)$ of colors. The sequence is of length $q$ for some $q=O(\Delta^2)$ and  each color in the sequence is an element in the set $\{0,\ldots,q-1\}$. Then, in $q$ phases the vertices try to get colored. In phase $i$ each uncolored vertex $v$ tries to get permanently colored with color $p_v(i)$. Assume that vertex $v$ tries to get permanently colored with color $c$ in some phase; $v$'s try is successful if and only if none of its $d2$-neighbors is already colored with color $c$ or also tries to get color $c$ in phase $i$. This test can be performed in two communication rounds as follows (similar to the color try in \Cref{sec:randAlg}): Each vertex $v$ sends its \emph{candidate color} (or its permanent color if it is colored) for the current phase to its neighbors. Thus each neighbor $w$ of $v$ receives the candidate colors (and permanent colors) of all of its direct neighbors. Node $w$ reports back to $v$ whether $v$'s candidate color is in conflict with any of the candidates (or permanent colors) of $w$'s neighbors (including $w$ itself). Furthermore $w$ can perform this check for all of its neighbors in the same round without any congestion. After the $q$ phases
we reduce the number of colors to $\Delta(G^2)+1$ in $O(\Delta^2)$ additional rounds by using \Cref{thm:delta2_simple_color}.

\medskip 
The rest of the section is devoted to showing how vertices choose their sequence and why every vertex is colored with a color in $[q]$ after the $q$ phases. 
We first explain how vertex $v$ defines its sequence assuming that a $10\Delta^4$-coloring $\psi$ of $G^2$ is already given to the nodes. Recall that we compute such a coloring in $O(\Delta+\logstar n)$ rounds with \Cref{thm:Liniald2}. Afterwards we prove that every node is colored at the end of the process.
To determine their sequences vertices pick a common prime number $q$ with 
    $4\Delta^2<q<8\Delta^2$~.
Such a prime always exists due to Bertrand's postulate and nodes can locally obtain the prime if they know $\Delta$.\footnote{As the algorithm can also deal with already colored neighbors one can remove the assumption that nodes need to know $\Delta$ with the standard exponential doubling technique. } Then, associate each color in $\psi$ with a distinct polynomial $p: \mathbb{F}_q\rightarrow \mathbb{F}_q$ with coefficients in $ \mathbb{F}_q$ of degree at most $1$ and let $p_v$ be the polynomial associated with $\psi(v)$. Note, that such an assignment is possible as there are $q^2\geq 16\Delta^4\geq 10\Delta^4$ polynomials over $\mathbb{F}_q$ of degree at most $1$.\footnote{One can obtain an explicit mapping, e.g., by setting $p_v(x)=a_v+b_v\cdot x$ with $a_v=\lfloor\psi(v)/q\rfloor$ and $b_v=\psi(v)\mod q$. } 
The sequence of $v$ is defined by evaluating $p_v$ at the points of $\mathbb{F}_q$, that is, $v$'s sequence is $p_v(0),\ldots, p_v(q-1)$.
\medskip 

We call a phase \emph{blocked} for vertex $v$ if any of its $d2$-neighbors tries the same color as $v$ in this phase or some $d2$-neighbor is already colored with the color that $v$ tries in the phase. Thus $v$ is colored with a color in $[q]$ as soon as one phase is not blocked for $v$. We now upper bound the number of blocked phases.
\begin{lemma}
\label{lem:blockedPhases}
 Any vertex $v$ has at most $2\Delta^2$ blocked phases. 
\end{lemma}
\begin{proof}
Let $u$ and $v$ be $d2$-neighbors. As $u$ and $v$ have different colors in $\psi$ they choose distinct polynomials $p_u$ and $p_v$. As $p_u$ and $p_v$ are non-equal polynomials of degree $1$ over a prime field there is at most 1 $i\in \mathbb{F}_q$ with $p_v(i)=p_u(i)$. Thus any $d2$-neighbor $u$ of $v$ can cause at most one blocked phase for $v$ while $u$ is still trying to get colored. Once $u$ is permanently colored with some color $c_u$, we again have $p_v(i)=c_u$ for at most one $i$, as $p_v$ and $c_u$ are both non-equal polynomials of degree at most one; note that $p_v\neq c_u$ because if $p_v$ is a polynomial of degree $1$ it is different from the constant polynomial $c_u$, if $p_v$ is a polynomial of degree $0$, i.e., a constant $c_v$ it is different from $c_u$ as $u$ can only choose $c_u$ if $c_u\neq p_v(i)=c_v$ for some phase $i$.  Thus any $d2$-neighbor $u$ can cause at most one blocked phase after it is permanently colored. 

In total any $d2$-neighbor $u$ can cause at most two blocked phases for $v$ and as there are at most $\Delta^2$ distinct $d2$-neighbors there are at most $2\Delta^2$ blocked phases for $v$.
\end{proof}
We can now show that the presented algorithm reduces an $O(\Delta^4)$ coloring to a $O(\Delta^2)$-coloring in $O(\Delta^2)$ rounds. 

\begin{theorem}
\label{thm:d2locIt}
There is a deterministic \congest algorithm that generates a $O(\Delta^2)$-coloring of $G^2$ in time $O(\Delta^2)$ given an $O(\Delta^4)$ coloring of $G^2$.
\end{theorem}

\begin{proof}[Proof of \Cref{thm:d2ColoringDelta}]
We first prove that every vertex is colored after the last phase, that the coloring is proper, and that it uses at most $O(\Delta^2)$ colors. 
A vertex $v$ obtains a final color in the first phase in which $v$ is not blocked. Due to \Cref{lem:blockedPhases} there are at most $2\Delta^2$ blocked phases for a fixed vertex $v$. The total number of phases is $>2\Delta^2$, that is, each vertex has a phase in which it is not blocked and obtains a final color. The coloring is proper as $d2$-neighbors cannot get colored with the same color in the same phase as the color trial is negative in such a round and if one of them got colored before the other the latter one has to choose a different color. The coloring uses at most $q=O(\Delta^2)$ colors as the polynomials are evaluated over $\mathbb{F}_q$. 

The $q$ phases of the locally iterative algorithm need $O(\Delta^2)$ rounds as every color trial can be implemented in $O(1)$ rounds.
\end{proof}

\begin{proof}[Proof of \Cref{thm:d2ColoringDelta}]
The result follows by executing the algorithms of \Cref{thm:Liniald2}, \Cref{thm:d2locIt} and \Cref{thm:delta2_simple_color} in this order. 
\end{proof}


\section{Concentration Bounds}
We state here the key concentration bounds that we use.

%
Before we continue with the details of our algorithm we state the following standard Chernoff bound that we utilize frequently in our proofs. 

\begin{proposition}[Chernoff]
Let $X_1, X_2, \ldots, X_n$ be independent Bernoulli trials,
$X = \sum_{i=1}^n X_i$, and $\mu = E[X]$. Then, for any $0 < \delta \le 1$,
\begin{align}
\label{eq:chernoff-upper}
\Pr[X \ge (1+\delta)\mu] & \le e^{-\mu \delta^2/3}     \\
\label{eq:chernoff-lower}
\Pr[X \le (1-\delta)\mu] & \le e^{-\mu \delta^2/2}\ .
\end{align}
Also,
\begin{equation}
    X = O(\mu + \log n), \text{w.h.p.}
    \label{eq:concentration}
\end{equation}
\label{P:chernoff}
\end{proposition}

We also use the following inequality:
\begin{equation}
(1-1/x)^{x-1} \ge 1/e, \text{ for any } x > 1.
\label{eq:inv-e}
\end{equation}

\section{Deferred Proofs of Section \ref{sec:randAlg}}

\subsection{Deferred Proofs of Section \ref{ssec:similarity}}
\label{app:similarity}

\textsc{Theorem \ref{T:similarity}:} 
\emph{Let $k \in \{3,6\}$.
Let $u,v$ be d2-neighbors. 
If $(u,v) \in H_{1-1/k}$ (i.e., if $|S_v \cap S_u| \ge (1-1/(2k)) c_{10} \log n$), then they share at least $(1-1/k) \Delta^2$ common d2-neighbors, w.h.p.,
while if $(u,v) \not\in H$, then they share fewer than $(1-1/(4k)) \Delta^2$ common d2-neighbors, w.h.p.}
\begin{proof}
We give the proof only for $k=3$, for readability.
We indicated above how the case $\Delta^2 \le c_{10}\log n$ is treated; we assume from now that $\Delta^2 \ge c_{10}\log n$.

Let $I_{uv} = G^2[u]\cap G^2[v]$ be the intersection of the d2-neighborhoods of $u$ and $v$. 
For each $w \in I_{uv}$, let $X_w$ be the indicator r.v.\ that $w$ is
selected into the random sample $S$ and let $X = \sum_{w \in I_{uv}} X_w$. Note that $\mu = E[X] = c_{10} (\log n)/\Delta^2 \cdot |I_{uv}|$.

First, suppose $|I_{uv}| \le \nicefrac{2}{3}\, \Delta^2$.
Then $\mu \le \nicefrac{2}{3}\, (c_{10}\log n)$, and
by (\ref{eq:chernoff-upper}), the probability that $u$ and $v$ are neighbors in $H$ is bounded by
  \[ \Pr[uv \in H] = \Pr[X \ge \nicefrac{5}{6}\, c_{10} \log n]
\le \Pr[X \ge \nicefrac{5}{4}\,\mu] \le e^{-\nicefrac{2c_{10}}{3\cdot 48}\, \log n} \le n^{-c_{10}/72}\ . \]
Thus, setting $c_{10}$ large enough implies that the first half of the claim holds.

Now, suppose $|I_{uv}| \ge \nicefrac{11}{12}\, \Delta^2$.
Then, $\mu \ge \nicefrac{11}{12}\, (c_{10}\log n)$.
By (\ref{eq:chernoff-lower}), the probability that $u$ and $v$ are non-neighbors in $H$ is bounded above by
  \[ \Pr[uv \not\in H] = \Pr[X < \nicefrac{5}{6}\, c_{10} \log n]
\le \Pr[X < \nicefrac{10}{11}\,\mu] \le e^{-\mu/(2\cdot 11^2)} \le e^{-c_{10}/(2\cdot 11 \cdot 12)\, \log n} = n^{-\Omega(c_{10})}\ . \]
Thus, setting $c_{10}$ large enough implies that the second half of the claim holds.
\end{proof}

\textsc{Lemma \ref{L:rand-nbor}:}
\emph{A multiset $R_u$ of independent uniformly random $H$-neighbors of node $u$ can be generated in $O(|R_u|+\log n)$ rounds.}
\begin{proof}
Each d2-neighbor's string $r_w$ gets forwarded (by $u'$) with probability $2^{c_{11}}\log n / \Delta^2$. Thus, the number  $y_{u}$ of strings that get forwarded to $u$ is expected $E[y_u] = 2^{c_{11}} \log n$. Setting $c_{11} \ge \max(4, 1+\log c_3)$.
By Chernoff (\ref{eq:chernoff-lower}), $u$ receives at least $2^{c_{11}-1} \log n \ge c_3 \log n$ strings with probability $1-1/n^2$.
Thus, with probability at least $1-1/n$, all nodes receive at least $c_3 \log n$ strings.
Hence, the node $u$ will correctly identify the $H$-neighbor whose bitstrings have the smallest XOR with its random string, which gives a uniformly random sampling.
The probability that some pair of nodes receive the same bitstring is at most $2^{-4\log n} = n^{-4}$. Thus, with probability at least $1-1/n^2$, there is always a unique node with the smallest XORed string.

Independence follows because a collection $\{r_w\}_w$ of uniformly random bitstrings that is XORed with a particular (not necessarily random) bitstring $b_u$ forming collection $\{ b_u \oplus r_w\}_w$ stays uniformly random.
The same holds then for the collection of strings $\{b_u\}_u$
that is XORed with a string $r_w$ of a fixed node.
\end{proof}

\subsection{Deferred Proofs of Section \ref{ssec:dense}}
\label{app:struct}
\textsc{Observation \ref{O:sparse}:} \emph{
Every live node is solid after Step 1 of \alg{d2-Color}, w.h.p.}
\begin{proof}
Let $v$ be a node of leeway $\phi \ge c_1 \Delta^2$ at the end of Step 2.
Then, in each iteration of the step, the color tried by $v$ has probability $\phi/\Delta^2 \ge c_1$ of being previously unused by d2-neighbors of $v$. Furthermore, the probability that no other d2-neighbor tries the same color in the same round is at least $(1-1/(\Delta^2+1))^{\Delta^2} \ge 1/e$, applying (\ref{eq:inv-e}). Thus,
with probability at least $\phi/(e \Delta^2) \ge c_1/e$, $v$ becomes colored in that round. Hence, the probability that it is not colored in all $c_0 \log n$ rounds is at most $(1-c_1/e)^{c_0\log n} \le e^{-c_0 c_1/e \log n} \le n^{-c_0 c_1/e} \le n^{-3}$, applying the bound $c_0 \ge 3e/c_1$.
So, with probability at least $1-1/n^2$, all nodes have leeway at most $c_1 \Delta^2$ after Step 2.

From the contrapositive of Prop.~\ref{P:sparsity} we obtain that the sparsity of $v$ is at most $\zeta \le 4e^3 \phi$, w.h.p.
\end{proof}

\textsc{Lemma \ref{L:h-degree}:} \emph{Let $v$ be a node of sparsity $\zeta$.
Then,
\begin{enumerate}
    \item $v$ has at least $\Delta^2 - 8\zeta/k -4/k$ neighbors in $H_{1-k}$, and
    \item The number of nodes that are within distance 2 of $v$ in $H$ but are not d2-neighbors of $v$ is
$|N_{H^2}(v) \setminus N_{G^2(v)}| \le 6\zeta$.
\end{enumerate}
}
 \begin{proof}
\textbf{1}. 
Let $k' = 1 - k/4$.
A d2-neighbor of $v$ that is not a $H_{1-k}$-neighbor can share at most $k' \Delta^2$ common d2-neighbors with $v$, by Thm.~\ref{T:similarity}, while $H_{1-k}$-neighbors can share up to $\Delta^2-1$ d2-neighbors with $v$.
In other words, the d2-neighbors of $v$ can have degree at most 
$\Delta^2-1$ ($k'\Delta^2$) in $G^2[v]$ if they are $H$-neighbors (non-$H$-neighbors), respectively.
The number of edges in $G^2[v]$ is then at most
\[ \frac{1}{2}\left( |N_{H_{1-k}}(v)| \Delta^2 + (\Delta^2 - |N_{H_{1-k}}(v)|) k' \Delta^2\right) = 
\frac{\Delta^2}{2} \left(k' \Delta^2 + |N_{H_{1-k}}(v)| \frac{k}{4}\right) \ . \]
By the definition of sparsity, the number of edges in $G^2[v]$ equals $\Delta^2((\Delta^2-1)/2 - \zeta)$.
Combining the two bounds, 
\[ |N_{H_{1-k}}(v)| \frac{k}{4} \ge \Delta^2 - 1 - 2\zeta - k'\Delta^2 
  = \Delta^2 \frac{k}{4}  - 1 - 2\zeta \]
Namely, the number of $H_{1-k}$-neighbors of $v$ is bounded below by $\Delta^2 - 8\zeta/k -4/k$.

\textbf{2}. By sparsity, there are $\nicefrac{1}{2}\, \Delta^2(\Delta^2 - 2\zeta)$ edges within $G^2[v]$. Thus, there are at most $2\zeta \Delta^2$ edges of $H$ that have exactly one endpoint in $N_{G^2}(v)$. Nodes in $N_{H^2}(v)$ share at least $\Delta^2 - \Delta^2/3 - \Delta^2/3 = \Delta^2/3$ d2-neighbors with $v$. Thus, there are at most $6\zeta$ nodes in $H^2[v]$ that are not in $G^2[v]$.
\end{proof}

\textsc{Lemma \ref{L:H-neighbors}:} \emph{Let $v$ be a solid node.
Then,
\begin{enumerate} 
  \item $v$ has at least $\Delta^2/2$ $H'$-neighbors.
  \item Every $\hat{H}$-neighbor of $v$ has at least $\Delta^2/3$ $H$-neighbors.
  \item The degree sum in $N_{H'}(v)$ is bounded below by
      \[ \sum_{u \in N_{H'}(v)} deg_H(u) \ge |N_{H'}(v)| (\Delta^2 - c_8\phi), \]
     for constant $c_8 \le 4000$.
\end{enumerate}}
\begin{proof}
\textbf{1}. Let $\zeta$ be the sparsity of $v$ and $\phi$ be its leeway, which satisfy $\phi \le c_1 \Delta^2$ and $\zeta \le 4e^3\phi$, since $v$ is solid.
By Lemma \ref{L:h-degree}(1), $v$ has at least $\Delta^2 - 48\zeta - 24 \ge \Delta^2 - 50\zeta$ neighbors in $\hat{H}$, where $\zeta$ is the sparsity of $v$. At most $\phi$ of those nodes have more than one 2-path to $v$, since $v$'s slack is at most $\phi$. Hence, it has at least $\Delta^2 - 50\zeta - \phi \ge \Delta^2(1 - 201 e^3 c_1) \ge \Delta^2/2$ $\hat{H}$-neighbors with a single 2-path to $v$, using that $c_1 \le 1/(402 e^3)$.

\textbf{2}. 
Let $v$ be a solid node and $u$ a node in $\hat{H}[v]$. Let $X$ be the set of nodes in $G^2[v]$ that share at least $2\Delta^2/3$ d2-neighbors of $G^2[v]$ with $u$, and let $Y$ be the set of d2-neighbors of $u$ in $G^2[v]$.
We want to show that $|X \cap Y| \ge \Delta^2/3$.

Since a node in $\hat{H}[v]$ shares at least $5\Delta^2/6$ d2-neighbors with $v$, any pair of nodes in $\hat{H}[v]$ share at least $\Delta^2 - 2\Delta^2/6 = 2\Delta^2/3$ d2-neighbors in $G^2[v]$. Namely, $|X| \ge |N_{\hat{H}}(v)|$, and by Lemma \ref{L:H-neighbors}(1), $|N_{\hat{H}}(v)| \ge \Delta^2/2$. 
Since $u$ is in $\hat{H}[v]$, $|Y| \ge 5\Delta^2/6$.
Thus, $|X \cap Y| \ge |X| - |N_{G^2}(v)\setminus Y| \ge \Delta^2/2 - (1-5/6)\Delta^2 = \Delta^2/3$.

\textbf{3}. 
Since $v$ is solid, it has sparsity $\zeta \le 4e^3 \phi$. Thus, $G^2[v]$ contains $\binom{\Delta^2}{2} - \zeta \Delta^2$ edges. 
By Lemma \ref{L:h-degree}(1), $v$ has degree $\Delta^2 - 48\zeta - 24$ in $\hat{H}$. The at most 
$48\zeta+24$ nodes in $N_{G^2}(v) \setminus N_{\hat{H}}(v)$ have degree sum at most $(48\zeta+24) \Delta^2$. Thus, the number of edges in $\hat{H}[v]$ is at least  
$\binom{\Delta^2}{2} - (49\zeta +24)\Delta^2 \ge \binom{\Delta^2}{2} - (196e^3\phi+24)\Delta^2$.
Recall that at most $\phi$ nodes in $N_{\hat{H}}(v)$ can have multiple paths to $v$. 
Then, $H'[v]$ has at most $\phi \Delta^2$ fewer edges than $\hat{H}[v]$, which is still at least $\binom{\Delta^2}{2} - c_8\phi\Delta^2$, for $c_8 = 196e^3+2 \le 4000$. A pair of nodes in $H'[v]$ have at least $2\Delta^2/3$ d2-neighbors of $v$ in common, since they each have at least $5\Delta^2/6$ d2-neighbors in common with $v$. Thus each edge in $H'[v]$ connects $H$-neighbors. In other words, nodes in $H'[v]$ have $\Delta^2 - c_8 \cdot \phi$ $H$-neighbors, on average.
\end{proof}

\subsection{Deferred Proof of Section \ref{ssec:correctness}}
\label{app:correctness}

 \textsc{Lemma \ref{L:survival}}: \emph{Let $v$ be an active live node. 
Any given query sent from $v$ towards a node $w$ via a node $u\in H'\subseteq \hat{H}[v]$ survives with constant probability at least $c_6 \ge 1/7$, independent of the path that the query takes.}

\begin{proof}
We fix a particular query $Q$ that has been generated and we use the following notation for nodes on its way (in the case the query does not get dropped) $v-u'-u-w'-w$.
For the intermediate nodes $u'$ and $w'$, queries are only dropped if they are to continue towards the same next destination (i.e., to the same $u$ or $w$). We will account for their dismissal there. Namely, we have $v$  assign each query a random priority, and $u$ retains the received query with the highest priority. A drop at $u'$ is therefore also a drop at $u$.

At the remaining nodes, $u$ and $w$, we bound the expected number of queries arriving -- other than $Q$ -- from above by some constant $c$. By Markov inequality, the number of other queries received is then at most $2c$, with probability at least $1/2$. The probability of surviving the culling with $2c$ other competitors is $1/(2c+1)$. 
Thus, with probability at least $1/(2(2c+1))$, the query $Q$ survives this stage.


\textit{Being dropped at $u$:} After the second round, the node $u$ has at most $(24e^3+1)\phi$ live $H$-neighbors, since $v$ has at most $\phi$ live d2-neighbors and by Lemma \ref{L:h-degree}(2) and Obs.~\ref{O:sparse}, 
there are at most $24e^3\phi$ live nodes that are $H^2$-neighbors of $v$ but not $G^2$-neighbors of $v$. $u$ receives a query from each with probability $1/(6000\phi)$, for an expected at most $(24e^3+1)\phi/(6000\phi) \le 1/12$ queries.

A query can also be dropped in Step 2 if it arrives at a node in $\hat{H}[v] \setminus H'$, but the lemma is conditioned on queries sent towards nodes in $H'$, i.e., $u\in H'$. 


\textit{Being dropped at $w$:} Finally, we consider a node $w$ at the end of round 2 of Step 4. $w$ has at most $\Delta^2$ $H$-neighbors. Each of its $H$-neighbors with a live $\hat{H}$-neighbor has $H$-degree at least $\Delta^2/3$ by Lemma \ref{L:H-neighbors}(2), 
and has expected at most $1/12$ query to send to a random $H$-neighbor. 
Hence, the expected number of other queries that $w$ receives is at most $3/12=1/4$. 

\textit{Combined probability of being dropped:} The probability of survival is at least $c_6 \ge 1/(2(1/6+1)) \cdot 1/(2(1/2+1)) \cdot = 3/21 = 1/7$.
\end{proof}

 \textsc{Lemma \ref{L:proposal-expected}}: \emph{ 
 An active live node receives at most one proposal in expectation. This holds even in the setting where no queries are dropped. }
 \begin{proof}
  Let $P_v$ be the set of nodes that are $H$-neighbors of $H'$-neighbors of $v$ but are not d2-neighbors of $v$. These are the only nodes that generate proposals for $v$ in Step 5. Each $H'$-neighbor $u$ of $v$ receives a query from $v$ with probability $1/(6000\phi)$ and sends it to a random $H$-neighbor. $u$ has at least $\Delta^2/3$ $H$-neighbors by Lemma \ref{L:H-neighbors}(2). Thus, the probability that a given node in $P_v$ receives a query involving $v$ in Step 5 is at most $1/(6000\phi) \cdot \Delta^2 \cdot 3/\Delta^2 = 3/(6000\phi)$. By Lemma \ref{L:h-degree}(2) and Obs.~\ref{O:sparse}, $P_v$ contains at most $24 e^3 \phi$ nodes. Hence, the expected number of proposals $v$ receives from Step 5 is $24 e^3\phi \cdot 3/(6000\phi) = 72e^3/6000 \le 1/4$.

We next bound the expected number of proposals due to Step 3.
The probability $p_u$ that a given $H'$-neighbor $u$ of $v$ picks a color that is not used among its $H$-neighbors, over the random color choices, is 
\[p_u = (\Delta^2 - deg_{H}(u))/deg_{H}(u) \le 3(1 - deg_H(u)/\Delta^2)\ , \] 
where we used Lemma \ref{L:H-neighbors}(2) in the inequality.
The probability that a random query from $v$ leads to a color proposal is then 
  \[ \frac{\sum_{u \in H'[v]} p_u}{|N_{H'}(v)|} \le \frac{1}{|N_{H'}(v)|\Delta^2} \sum_{u \in H'[v]} 3 \left(\Delta^2 - deg_H(u)\right) \le  \frac{c_8 \phi}{|N_{H'}(v)|}\ ,  \]  
applying Lemma \ref{L:H-neighbors}(3). The expected number of queries sent from $v$ to $H'$-neighbors to $N_{H'}$ is $|N_{H'}(v)|/(6000\phi)$, so the expected number of proposals that $v$ receives is $c_8/6000 \le 2/3$.
\end{proof}

\textsc{Lemma \ref{L:proposal-tried}.} \emph{Let $v$ be an active live node.
  Conditioned on the event that a particular color is proposed to $v$, the probability that $v$ \emph{tries} the color is at least $c_9 \ge 1/6$.}
\begin{proof}
By Lemma \ref{L:proposal-expected}, $v$ receives at most one proposal in expectation.
By Markov's inequality, $v$ receives at most two proposals, with probability at least $1/2$. Let $m$ be a specific proposal. 
Conditioned on $m$ being proposed, $v$ receives at most three proposals, with probability at least $1/2$. Hence, with probability at least $1/6$, $m$ will be the proposal selected to be tried.
\end{proof}

\subsection{Deferred Proof of Section \ref{ssec:improved}}
\label{app:improved}

\textsc{Theorem \ref{T:learnpalette}:} 
\emph{The time complexity of \alg{LearnPalette}($\varphi$) with $\varphi=O(\log n)$ is $O(\log n)$, when $\Delta = \Omega(\log n)$.}
\begin{proof}
Each live node has $\Omega(\Delta^2)$ neighbors and selects $Z$ of them uniformly to become handling nodes. Thus, each node has probability $O(Z/\Delta^2)$ of becoming a handling node for a given live d2-neighbor, and since it has $O(\varphi)$ live d2-neighbors (by the precondition), it becomes a handling node for an expected $O(\varphi \cdot Z/\Delta^2)$ live nodes. 
\begin{itemize}
\item The flooding in Step 2 takes as many rounds as there are live nodes in a immediate neighborhood, which is $O(\varphi)$.
\item In Step 3 involves sending a single message to each handling node (of each live node), which is easily done in $O(1)$ expected time, or $O(\log n)$ time, w.h.p.
\item In Step 4, each node forwards expected $P/\Delta$ messages from each handling immediate neighbor, and it has $O(Z\varphi/\Delta)$ such immediate neighbors. 
Thus, it sends out $O(P Z\varphi/\Delta^2)$ messages to random neighbors, or $O(P Z \varphi/\Delta^3)$ message per outgoing edge, w.h.p. This takes time $O(PZ \varphi/\Delta^3)$. 
\item In Step 5, a colored node needs to forward its color to the handling node $z_v^i$ of each of its $O(\varphi)$ live d2-neighbors. Since it is sent along $\Delta^2/P \log n$ paths, and due to the conductance,  w.h.p., the color reaches a node in $Z_v^i$ (the set of nodes informed of $z_v^i$). 

The path from a colored node $u$ to a handler $z_v^i$ for a live node has two parts: the path $p_{u,w}$ from $u$ to an node $w$ that knows the path to $z_v^i$, and path $q_{w,z_v^i}$ from $w$ to $z_v^i$. 
A given node $a$ has probability $1/\Delta^2$ of being an endpoint of a given path $p_{u,w}$ (from a d2-neighbor); there are $O(\varphi)$ live d2-neighbors (by the precondition, weaker form), $\Delta^2$ colored d2-neighbors, and $\Delta^2/P \log n$ copies of messages sent about each. Thus, a given node has expected 
\[ O(\varphi \cdot \Delta^2 \cdot \Delta^2/P \log n \cdot 1/\Delta^2) = O(\varphi/P \cdot \Delta^2 \log n) \] 
random paths going to it.
Similar argument holds for a node being an intermediate node on a path $p_{u,w}$: the number of immediate neighbors goes to $\Delta$, while the probability of being a middle point on the path goes to $1/\Delta$, resulting in the same bound.
Thus, the load on each edge is $O(\Delta \varphi/P \log n)$. 

For the $q$-paths, the main congestion is going into the handler. Observe that there are only $O(\log n)$ $p$-paths that reach an informed node $w$. Hence, the number of paths going into a given handler is the product of the size of the block, times $\log n$: $\Delta^2/Z \cdot \log n)$. So, the load on an incoming edge into a handler is $O(\Delta/Z \log n)$, w.h.p. (when $Z = O(\Delta)$).

\item The pipelining of Steps 6-7 takes $O(|T_v|)$ rounds, and by Lemma \ref{l:last}, $|T_v| = O(\log n)$, w.h.p. 

\end{itemize}

To summarize, the dominant terms of the time complexity are $O(\log n)$ (Steps 2, 6-7), $O(PZ\varphi/\Delta^3) = O(PZ (\log n)/\Delta^3)$ (Step 4), $O(\Delta\varphi/P \log n) = O(\Delta(\log n)^2/P)$ (first half of Step 5), and $O(\Delta/Z \log n)$ (second half of Step 5).

Optimizing, we set $Z = \Delta$ and $P = \Delta \sqrt{\Delta\log n}$, 
for time complexity of $O(\log n (1 + \sqrt{(\log n)/\Delta}))$, which is $O(\log n)$ when $\Delta = \Omega(\log n)$. 
\end{proof}

\end{document}